\documentclass{article}
\usepackage{mathtools, authblk}

\usepackage[margin=1in,hmargin=1in]{geometry}
\usepackage{booktabs}
\usepackage{enumerate}
\usepackage{tikz}
\usepackage{float}
\usepackage{url}
\usetikzlibrary{matrix}
\usepackage{natbib}
\bibpunct{(}{)}{;}{a}{}{;}
\usepackage{float}
\usepackage[caption = false]{subfig}
\usepackage{graphicx}
\usepackage{subfig}

\usepackage{amsmath,amssymb,amsthm}
\newtheorem{theorem}{Theorem}
\newtheorem{lemma}{Lemma}
\newtheorem{remark}{Remark}

\title{Joint Structural Break Detection and Parameter Estimation in High-Dimensional Non-Stationary VAR Models}
\begin{document}

\author[ ]{Abolfazl Safikhani\thanks{as5012@columbia.edu}}
\author[ ]{Ali Shojaie\thanks{ashojaie@uw.edu}}
\affil[ ]{Columbia University and University of Washington}

\maketitle


\vspace{1cm}

\begin{abstract}
Assuming stationarity is unrealistic in many time series applications. A more realistic alternative is to assume piecewise stationarity, where the model is allowed to change at potentially many time points. We propose a three-stage procedure for consistent estimation of both structural change points and parameters of high-dimensional piecewise vector autoregressive (VAR) models. In the first step, we reformulate the change point detection problem as a high-dimensional variable selection one, and solve it using a penalized least square estimator with a total variation penalty. We show that the proposed penalized estimation method over-estimates the number of change points. We then propose a selection criterion to identify the change points. In the last step of our procedure, we estimate the VAR parameters in each of the segments. We prove that the proposed procedure consistently detects the number of change points and their locations. We also show that the procedure consistently estimates the VAR parameters. The performance of the method is illustrated through several simulation studies and real data examples.

\end{abstract}

\section{Introduction}\label{sec:intro}
Emerging applications in biology \citep{smith2012future, fujita2007modeling, mukhopadhyay2006causality} and finance \citep{de_2008, fan_2011sparse} have sparked an interest in methods for analyzing high-dimensional time series. Recent work includes new regularized estimation procedures for vector autoregressive (VAR) models \citep{Basu_2015, matteson_2017}, high-dimensional generalized linear models \citep{hall_2016} and high-dimensional point processes \citep{HansenETAL_2015,  ChenWittenShojaie_2017}. 
Related methods have also focused on joint estimation of multiple time series \citep{QiuETAL_2016}, estimation of (inverse) covariance matrices \citep{xiao2012covariance, chen2013covariance, tank2015bayesian}, and estimation of high-dimensional systems of differential equations \citep{LuETAL_2011, ChenShojaieWitten_2016}.

Despite considerable progress on both computational and theoretical fronts, the vast majority of existing work on high-dimensional time series assumes that the underlying process is \emph{stationary}. However, multivariate time series observed in many modern applications are \emph{nonstationary}. For instance, \citet{ClaridaGaliGertler_2000} show that the effect of inflation on interest rates varies across Federal Reserve regimes. Similarly, as pointed out by \citet{OmbaoVonSachsGuo_2005}, electroencephalograms (EEGs) recorded during an epileptic seizure display amplitudes and spectral distribution that vary over time. This nonstationarity in EEG signals is illustrated in Figure~\ref{fig_EEG_full}, which shows the signals recorded at 18 EEG channels during an epileptic seizure from a patient diagnosed with left temporal lobe epilepsy \citep{OmbaoVonSachsGuo_2005}. 
The sampling rate in this data is 100~Hz and the total number of time points per EEG is $T =22,768$ over $\sim$228 seconds. Based on the neurologist's estimate, the seizure took place at $t = \sim85s$. Figure~\ref{fig_EEG_full} also suggests that the magnitude and the variability of EEG signals change  around that time. 

\begin{figure}[t]
\begin{center}
\includegraphics[width=0.5\linewidth, clip=TRUE, trim=0cm 0cm 0cm 2cm]{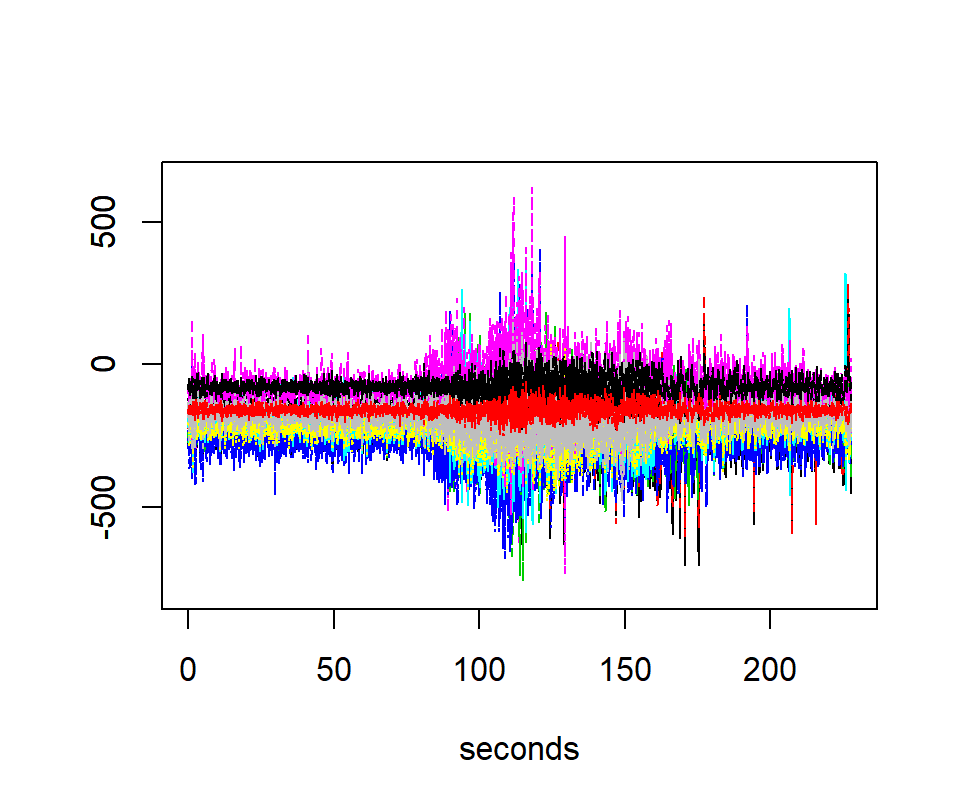}
\vspace{-0.5cm}
\caption{EEG signals from a patient diagnosed with left temporal lobe epilepsy. The data was recorded at 18 locations on the scalp during an epileptic seizure over 22,768 time points.}\label{fig_EEG_full}
\end{center}
\end{figure} 

Detecting structural break points in high-dimensional time series and obtaining reliable estimates of model parameters are important from multiple perspectives. First, structural breaks often reveal important changes in the underlying system and are, hence, scientifically important. In our EEG example, automatic detection of structural breaks can assist clinicians in identifying seizures. 
Second, changes in model parameters before and after break points often provide important scientific insight. For instance, the occurrence of epileptic seizure is expected to change the mechanism of interactions among brain regions. Such changes can be seen in Figure~\ref{fig_EEG_network}. The figure shows networks of interactions among EEG channels before and after seizure. These networks are obtained from estimates using our proposed method, as described in Section~\ref{sec:data}. Briefly, edges in the first two networks correspond to \emph{Granger causal} relations \citep{granger1969} among EEG channels before and after the period of seizure; the occurrence of seizure is also automatically detected using our proposed method. It can be seen that while the two networks share many edges, they also exhibit important differences. Perhaps most notable are changes in connectivity patterns of channels T5, P3 and Pz, which measure brain activity in the left temporal lobe, the cite of epilepsy in the patient. Without reliable estimates of model parameters, gaining such scientific insight may not be feasible. 
Third, identifying structural breaks in time series is also crucial for proper data analysis. The last network in Figure~\ref{fig_EEG_network} shows the network of interactions obtained from the full EEG data, ignoring the structural break due to seizure. This network is much more dense than the other two, and is, in fact, close to a fully connected network, which is rather unexpected. This example underscores that ignoring the structural breaks and assuming stationarity when analyzing times series can result in severe estimation bias.  

In this paper, we develop a regularized estimation procedure to simultaneously  detect the structural break points, and estimate the model parameters in high-dimensional piecewise stationary VARs with possibly many break points. 
{
We show that our proposed three-stage procedure is consistent for identifying the number and location of structural breaks in the covariance structure of multivariate time series, and for estimating the model parameters. 
To the best of our knowledge, this is the first method that can simultaneously identify the change points and estimate the model parameters in high-dimensional nonstationary time series with infinitely many break points. In fact, while change point detection in multivariate time series has been an active area of recent research (see the review in Section~\ref{sec:othermethods}), existing procedures do not provide consistent estimates of model parameters. This is in part due to the difficulty of change point detection, especially in high dimensions. More specifically, high-dimensional change point detection methods, including ours, only provide consistent estimates of the \emph{relative location} of change points. Thus, a na\"ive approach based on using the estimated change points to break the time series into `stationary' pieces, and then using high-dimensional parameter estimation procedures, may not lead to consistent estimates of model parameters. This point is formalized at the end of Section~\ref{sec:estimation} (Remark~5) and is illustrated in our simulation studies, where it can be seen that combining the state-of-the-art change point detection procedure with regularized  parameter estimation methods results in inferior estimates of model parameters. 
The third step of our procedure circumvent this difficulty by utilizing our improved rates of consistency of change point detection, and dividing the time series into segments that are guaranteed to be stationary.  

Another advantage of our method is that it can handle real high-dimensional VAR models with infinitely many structural breaks, where the number of time series components can grow exponentially fast with respect to the number of time points. 
This is in contrast to most existing methods for change point detection, which can only handle univariate or multivariate processes where the number of time series components grows at most polynomially with the number of time points 
\citep[see,  e.g.,][]{Chan_2014, cho2015multiple}. The high-dimensionality of time series brings new challenges in deriving asymptotic results, compared to univariate or finite-dimensional cases. We derive rates of consistency that are better than the previously established rates, and require less stringent assumptions in some cases; see Section~\ref{sec:comparison} for detailed comparison of our method with existing procedures. 
}

The rest of this paper is organized as follows. Before describing the piecewise stationary model in Section~\ref{sec:model}, we review related procedures for nonstationary time series in Section~\ref{sec:othermethods}. An initial estimation procedure and its asymptotic properties are discussed in Section~\ref{sec:initial}. In Section~\ref{sec:consistency}, we show that under reasonable assumptions, structural breaks in high-dimensional VAR models can be consistently estimated. 
Our procedure for consistent parameter estimation is presented in Section~\ref{sec:estimation} and 
{
finite sample tuning parameters selection for the procedure are discussed in Section~\ref{sec:tuningselection}.} 
Results of simulation experiments are presented in Section~\ref{sec:sims}. In Section~\ref{sec:data}, we illustrate the utility of the proposed method by applying it to identify structural break points in two multivariate time series. We conclude the paper with a discussion in Section~\ref{sec:disc}. Technical lemmas, proofs, algorithmic details and additional simulation results are gathered in the Appendix. 

\begin{figure}[t!]
\begin{center}
\includegraphics[width=0.6\linewidth, clip=TRUE, trim=0cm 0cm 0cm 0.8cm]{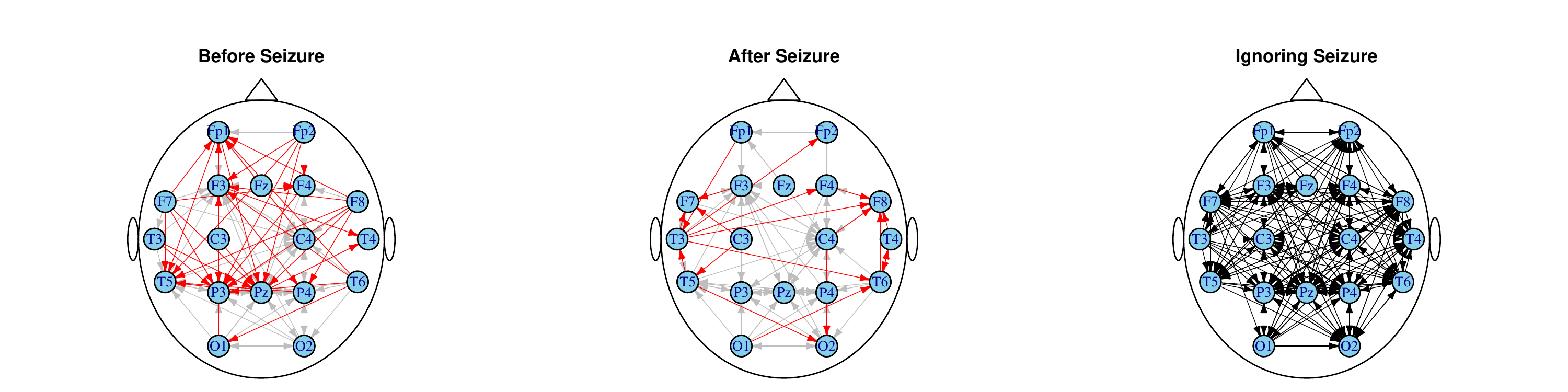}
\vspace{-0.5cm}
\caption{Network of Granger causal interactions among EEG channels based on data from Figure~\ref{fig_EEG_full}. The plots show the schematic locations of the EEG channels. The first two figures show interactions among EEG channels before and after the period of seizure. Gray edges in these two networks show common edges, while red edges show interactions identified either  before or after seizure. The rightmost network shows interactions from an estimate obtained by ignoring the structural break in the time series.}\label{fig_EEG_network}
\end{center}
\end{figure} 

\subsection{Related Methods for Nonstationary Time Series}\label{sec:othermethods}

Non-stationary VAR models have been primarily studied in univariate or low-dimensional settings. Existing approaches include models that fully parameterize the evolution of the transition matrices of time-varying VARs, or enforce a Bayesian prior on the structure of the time dependence \citep{Primiceri_2005}. 

A popular approach for analyzing nonstationary time series is to assume that the process is \emph{locally stationary}; local stationarity means that in each small time interval, the process is well-approximated by a stationary one. This notion has been studied in low dimensions by, e.g., \citet{Dahlhaus_2012}; \citet{SatoMorettinETAL_2007} proposed a wavelet-based method for estimating time-varying VAR  coefficients. 
Recently, \citet{DingQiuChen_2016} considered estimation of high-dimensional time-varying VARs by solving time-varying Yule-Walker equations based on kernelized estimates of auto-covariance matrices. 

Methods based on local stationarity are theoretically appealing and suitable in certain applications. However, local stationarity may not hold in many applications. For instance, in the above EEG example, assuming that the process can be locally approximated by a stationary one at the time of seizure may be unrealistic. A more natural assumption in such settings is that the process is \emph{piecewise stationary} --- i.e., the process is stationary in each of (potentially many) regions, e.g., before and after seizure. 

A number of methods have been proposed for analyzing univariate piecewise stationary time series. For instance, \citet{Davis_2006}, \citet{Chan_2014} and \citet{Bai_1997} propose different methods for identifying structural break points in univariate time series. Various methods have also been proposed for detecting changes in multivariate time series. One of the early contributions to this literature was the SLEX method of \citet{OmbaoVonSachsGuo_2005}, which uses time-varying wavelets to detect changes in the covariance structure of multivariate time series. The test procedure of \citet{aue2009break} addresses a similar problem in possibly nonlinear time series, whereas \citet{aue2017detecting} takes a functional data perspective in order to identify changes in the mean structure of multivariate time series. More recent approaches by \citet{cho2015multiple} and \citet{cho2016change} use variants of CUSUM statistic in order to detect structural break points in high-dimensional time series. 

Despite significant progress, existing multivariate approaches do not provide estimates of model parameters. For instance, to deal with the large number of time series, \citet{OmbaoVonSachsGuo_2005} apply a dimension reduction step, whereas \citet{aue2009break}, \citet{cho2015multiple} and \cite{cho2016change} use some variation of CUSUM statistic. As a result, these methods only estimate the structural break points. In other words, while existing methods can identify the structural breaks in the EEG example, they do not reveal mechanisms of interactions among brain regions, which is of key interest for understanding changes in brain function before and after seizure. As discussed in Section~\ref{sec:estimation}, consistent estimation of model parameters in high-dimensional piecewise stationary VAR models introduces new challenges. Addressing these challenges and providing interpretable estimates of parameters in high-dimensional piecewise stationary VAR models are two key contributions of the current paper.

\section{Piecewise Stationary VAR Models}\label{sec:model}
A piecewise stationary VAR model can be viewed as a collection of separate VAR models concatenated at multiple break points over the observed time period. More specifically, suppose there exist $ m_0 $ break points $ 0 = t_0 < t_1 < \cdots < t_{m_0} < t_{m_0 + 1} = T + 1 $ such that for $t_{j-1} \leq t < t_j, j = 1, \ldots, m_0 + 1$,
\begin{equation}\label{eq:model}
y_t = \Phi^{(1,j)} y_{t-1} + \cdots + \Phi^{(q,j)} y_{t-q} +  {\Sigma_j}^{1/2} \varepsilon_t.
\end{equation}
Here, $ y_t $ is the $p$-vector of observed time series at time $ t $; $ \Phi^{(l,j)} \in \mathbb{R}^{p \times p}$ is the (sparse) coefficient matrix corresponding to the $l$th lag of a VAR process of order $q$ during for the $j$th segment, where $ j = 1, \cdots, m_0+1 $; $ \varepsilon_t $ is a multivariate Gaussian white noise with independent components, and $ \Sigma_j$ is the covariance matrix of the noise for the $ j$th segment. To simplify the notations, we sometime denote the noise as $ \varepsilon_t $ without specifying the covariance $\Sigma_j$; however, throughout the paper, we allow for different covariance matrices in each segment. {
Note that the first few observations in each segment are, in fact, affected by the last few observations from the previous segment. Thus, the break points do not really divide the time series into stationary segments; hence, strictly speaking, model~\eqref{eq:model} is not piecewise stationary. This feature can, in general, lead to additional challenges when estimating the parameters, but is circumvented by the third step of our procedure discussed in Section~\ref{sec:estimation}.}

Our  goal is to detect the break points, $ t_j$, together with estimates of the coefficient parameters $ \Phi^{(l,j)}$ in the high-dimensional case, where $ p > T $. To this end, we generalize the change-point detection ideas of  \cite{Harchaoui_2010} and \cite{Chan_2014} to the multivariate, high-dimensional setting, and extend them to obtain consistent estimates of model parameters. More specifically, our estimation procedure utilizes the following linear regression representation of the VAR process
\begin{equation}\label{regression}
\begin{pmatrix} y_q^\prime \\  y_{q+1}^\prime \\ \vdots  \\ y_T^\prime \end{pmatrix} = \begin{pmatrix} y_{q-1}^\prime & \ldots & y_0^\prime &  & 0 & & \ldots &  & 0 & \\  y_{q}^\prime & \ldots & y_1^\prime & y_{q}^\prime & \ldots & y_1^\prime  & & & 0 &\\  & \vdots & & & & & \ddots & &   \\ y_{T-1}^\prime & \ldots & y_{T-q}^\prime & y_{T-1}^\prime & \ldots & y_{T-q}^\prime & \ldots &  y_{T-1}^\prime & \ldots &  y_{T-q}^\prime \end{pmatrix} \begin{pmatrix} \theta_1^\prime \\  \theta_{2}^\prime \\ \vdots  \\ \theta_n^\prime \end{pmatrix} + \begin{pmatrix} \varepsilon_q^\prime \\  \varepsilon_{q+1}^\prime \\ \vdots  \\  \varepsilon_T^\prime \end{pmatrix},
\end{equation}
where $ n = T - q + 1 $. Throughout the paper, the transpose of a matrix $ A $ is denoted by $ A^\prime $. Denoting $ \Phi^{(.,j)} =  \begin{pmatrix} \Phi^{(1,j)} & \ldots & \Phi^{(q,j)} \end{pmatrix} \in \mathbb{R}^{p \times pq} $, we set $  \theta_1 = \Phi^{(.,1)} $; for $ i = 2, 3, \ldots, n $, we let 
\begin{equation}
\theta_i =  \left\{
	\begin{array}{ll}
		\Phi^{(.,j+1)} - \Phi^{(.,j)}, & \mbox{when } i = t_j  \,\, \mbox{for some} \, j  \\
		0, & \mbox{otherwise}.
	\end{array}
\right.
\end{equation} 
Note that in this parameterization, $ \theta_i \neq 0 $ for $ i \geq 2 $ implies a change in the VAR coefficients. Therefore, for $j = 1, \ldots, m_0$, the structural break points $ t_j$ can be estimated as time points $ i \geq 2 $, where $ \theta_i \neq 0 $. 

Noting that \eqref{regression} is a linear regression of the form
$$ 
\mathcal{Y} = \mathcal{X}  \Theta + E, 
$$
we rewrite it in vector form as 
\begin{equation}\label{eqn:vectorform}
\textbf{Y} = \textbf{Z} \mathbf{\Theta} + \textbf{E}, 
\end{equation}
where $ \textbf{Y} = \mbox{vec}(\mathcal{Y}) $, $ \textbf{Z} = I_p \otimes \mathcal{X} $, and $ \textbf{E} = \mbox{vec}(E) $, {
with $ \otimes $ denoting the tensor product of two matrices}. Let $ \pi = n p^2 q $. Then, $ \textbf{Y} \in \mathbb{R}^{np \times 1} $, $ \textbf{Z} \in \mathbb{R}^{np \times \pi} $, $ \mathbf{\Theta} \in \mathbb{R}^{\pi \times 1} $, and $ \textbf{E} \in \mathbb{R}^{np \times 1} $.

\section{An Initial Estimator}\label{sec:initial}
The linear regression representation in \eqref{eqn:vectorform} suggests that the model parameters $\mathbf{\Theta}$ can be estimated via regularized least squares. The regularization is necessary to both handle the growing number of parameters corresponding to potential change points, as well as the number of time series $p$. A simple initial estimate of parameters $\mathbf{\Theta}$ can thus be obtained by using an $\ell_1$-penalized least squares regression of the form  
{
\begin{equation}\label{eq_estimation}
\widehat{\mathbf{\Theta}} = \mbox{argmin}_{\mathbf{\Theta}} \frac{1}{n} \| \textbf{Y} - \textbf{Z} \mathbf{\Theta} \|_2^2 + \lambda_{1,n}  \| \mathbf{\Theta} \|_1 + \lambda_{2,n} \sum_{k=1}^{n} \left \| \sum_{j=1}^{k} \theta_j \right \|_1.
\end{equation}
}
Problem~\eqref{eq_estimation} uses a fused lasso penalty \citep{tibshirani2005sparsity}, with two $\ell_1$ penalties controlling the number of break points and the sparsity of the VAR model. This problem is convex and can be solved efficiently. 
{
In fact, by Proposition~1 of \citet{friedman2007pathwise}, problem \eqref{eq_estimation} can be solved by first finding a solution, $\widetilde{\Theta}^{(0)}$, for $\lambda_{2,n} = 0 $ and then applying element-wise soft-thresholding to $\widetilde{\Theta}^{(0)}$ in order to obtain the final estimate for $\lambda_{2,n} \ne 0 $. Details of this algorithm are described in Appendix~C.} 

Despite its convenience and computational efficiency, estimates from  \eqref{eq_estimation} do not correctly identify the structural break points in the piecewise VAR process. In fact, our theoretical analysis in the next section shows that the number of estimated break points from \eqref{eq_estimation}, i.e., the number of nonzero $ \widehat{\theta}_i \neq 0 $, $ i \geq 2 $, over-estimates the true number of break points. This is because the design matrix $ \textbf{Z} $ may not satisfy the restricted eigenvalue condition \citep{BickelETAL_2009} that is needed for establishing consistent estimation of parameters. However, as we show in Section~\ref{sec:assumptions}, the model from \eqref{eq_estimation} does achieve prediction consistency. In Section~\ref{sec:consistency} we show that this initial estimator can be refined in order to obtain consistent estimates of structural break points.


\subsection{Asymptotic Properties}\label{sec:assumptions}

Denote the set of estimated change points from \eqref{eq_estimation} by 
$$
	\widehat{\mathcal{A}}_n = \left\lbrace i \geq 2 :  \widehat{\theta}_i \neq 0 \right\rbrace. 
$$ 
The total number of estimated change points is then the cardinality of the set $ \widehat{\mathcal{A}}_n $; denote $ \widehat{m} = | \widehat{\mathcal{A}}_n | $. Let $ \widehat{t}_j$ be the estimated break points for $j = 1, \ldots, \widehat{m}$. Then, the relationship between $ \widehat{\theta}_j $ and $ \widehat{\Phi}^{(.,j)}$ in each of the estimated segments can be seen as:
\begin{equation}\label{eqn:psi}
\widehat{\Phi}^{(.,1)} = \widehat{\theta}_1, \hspace{1cm} \mbox{and} \hspace{1cm} \widehat{\Phi}^{(.,j)} = \sum_{i=1}^{\widehat{t}_j} \widehat{\theta}_i, \hspace{1cm} j = 1, 2, \ldots, \widehat{m}. 
\end{equation}

In this section, we show that $ \widehat{m} \geq m_0 $. We also show that there exist $ m_0 $ points within $\widehat{\mathcal{A}}_n$ that are `close' to the true break points. To this end, we first establish the prediction consistency of the estimator \eqref{eq_estimation}. Using a more careful analysis, we then show that the penalized least squares in \eqref{eq_estimation} identifies a larger set of \emph{candidate} break points. These result justify the second step of our estimation procedure described in the Section~\ref{sec:consistency}, which searches over the break points in $\widehat{\mathcal{A}}_n$ to find an optimal set of break points. 

Before stating our assumptions, we first define a few notations. 
Denote the number of nonzero elements in the $ k$-th row of $ \Phi^{(.,j)} $ by $ d_{kj} $, $ k = 1, 2, \ldots, p $ and $ j = 1, 2, \ldots, m_0 $. Further, for each $ j = 1, 2, \ldots, m_0 + 1 $ and $ k = 1, \ldots, p $, denote by $ \mathcal{I}_{kj} $ the set of all column indexes of $ \Phi_k^{(.,j)} $ at which there is a nonzero term, where $ \Phi_k^{(.,j)} $ denotes the $ k$-th row of $ \Phi^{(.,j)} $. Let $ \mathcal{I} = \cup_{k,j} \mathcal{I}_{kj} $, and define $ d_n = \max_{1 \leq k \leq p, 1 \leq j \leq m_0+1} |\mathcal{I}_{kj}| $. Further, let $ d_{n}^\star = \sum_{j=1}^{m_0+1} \sum_{k=1}^{p} d_{kj} $ be the total sparsity of the model. 
{
Note that our theoretical analysis concerns the high-dimensional case with many break points, where $p$, $m_0$ and the network sparsity increase with the number of time points, $T$. More specifically, $p \equiv p (n)$ and $m_0 \equiv m_0 (n)$ and $d_{kj} \equiv d_{kj}(n)$, where $n = T-q+1$. To simplify the notation, we suppress the $n$-index.}

\begin{itemize}
\item[A1] For each $ j =1, 2, \ldots, m_0+1 $, the process $ y_t^{(j)} = \Phi^{(1,j)} y_{t-1}^{(j)} + \cdots + \Phi^{(q,j)} y_{t-q}^{(j)} + {\Sigma_j}^{1/2} \varepsilon_t $ is a stationary Gaussian time series. Denote the covariance matrices $ \Gamma_j (h) = \mbox{cov} \left( y_t^{(j)}, y_{t+h}^{(j)} \right) $ for $ t, h \in \mathbb{Z} $. Also, assume that for $ \kappa \in [-\pi, \pi] $, the spectral density matrices $ f_j (\kappa) = {(2\pi)}^{-1} \sum_{l \in \mathbb{Z}} \Gamma_j (l) e^{-\sqrt{-1} \kappa l} $ exist; further 
$$ \max_{1 \leq j \leq m_0+1} \mathcal{M}(f_j) = \max_{1 \leq j \leq m_0+1} \left( \mbox{ess sup}_{\kappa \in [-\pi, \pi]} \Lambda_{\mbox{max}} (f_j(\kappa)) \right) < + \infty, 
$$ 
and 
$$ \min_{1 \leq j \leq m_0+1} \textbf{m}(f_j) = \min_{1 \leq j \leq m_0+1}  \left( \mbox{ess sup}_{\kappa \in [-\pi, \pi]} \Lambda_{\mbox{min}} (f_j(\kappa)) \right) > 0, 
$$ 
where $ \Lambda_{\mbox{max}}(A) $ and $ \Lambda_{\mbox{min}}(A) $ are the largest and smallest eigenvalues of the symmetric or Hermitian matrix $ A $, respectively.
\item[A2] The matrices $ \Phi^{(.,j)} $ are sparse. More specifically, for all  $ k = 1, 2, \ldots, p $ and $ j = 1, 2, \ldots, m_0 $, $ d_{kj} \ll p $,   {
i.e., $ d_{kj}/p = o(1) $}. Moreover, there exists a positive constant $ M_\Phi > 0  $ such that 
$$  
\max_{1 \leq j \leq m_0+1} \left\| \Phi^{(.,j)}  \right\|_\infty  \leq M_\Phi. 
$$
\item[A3] There exists a positive constant $ v $ such that 
{
$$ 
\min_{1 \leq j \leq m_0}   \left\| \Phi^{(.,j+1)} -  \Phi^{(.,j)} \right\|_2   \geq v > 0. 
$$
}
Moreover, there exists a vanishing positive sequence $ \gamma_n  $ such that, as $ n \rightarrow \infty $,
\[
\hspace{-1.5cm}  \min_{1 \leq j \leq m_0+1} | t_j - t_{j-1} | / (n\gamma_n) \rightarrow + \infty, \, , {
\,\,\, \mbox{and} \,\,\, d_{n}^\star  \sqrt{\frac{\log p}{n \gamma_n}} \rightarrow 0.} 
\] 
\end{itemize}

Assumption~A1 allows us to obtain appropriate probability bounds in high dimensions. This assumption does not restrict the applicability of the method since it is valid for large families of VAR models \citep{Basu_2015}. The second part of A1 will also be needed in the proof of consistency of VAR parameters once the break points are detected. Assumption~A2 is related to the total sparsity of the model. The sequence $ \gamma_n $ is directly related to the detection rate of the break points $t_j; j = 1, \ldots, m_0$. Assumption~A3 connects this detection rate to the tuning parameter chosen in the estimation procedure. Also, this assumption puts a minimum distance-type requirement on the coefficients in different segments. This can be regarded as the extension of Assumption~H2 in \citet{Chan_2014} for univariate time series to the high-dimensional case. 

As pointed out earlier, and discussed in \cite{Chan_2014} and \cite{Harchaoui_2010}, the design matrix $\textbf{Z}$ in \eqref{eq_estimation} may not satisfy the restricted eigenvalue condition needed for parameter estimation consistency \citep{BickelETAL_2009}.  Thus, as a first step towards establishing the consistency of the proposed procedure, we next establish the prediction consistency of the estimator from \eqref{eq_estimation}. 

\begin{theorem}\label{thm_pred_error}
Suppose A1 and A2 hold. Choose $ \lambda_{1,n} =  2 C \sqrt{\frac{\log(n) + 2\log(p) + \log(q)}{n}} $ for some $ C  > 0 $ {
and $ \lambda_{2,n} = o\left( (n d_n^\star)^{-1} \right) $}, and assume $ m_0 \leq m_n $ with $ m_n = o \left( \lambda_n^{-1} \right) $. Then, with high probability approaching 1 as $ n \to +\infty $, 
\begin{equation}\label{eq:predconsistency}
\frac{1}{n} \left\| \textbf{Z} \left( \widehat{ \mathbf{\Theta} } - \mathbf{\Theta}  \right)  \right\|_2^2 \leq 2 M_\Phi \lambda_n m_n   \max_{1 \leq j \leq m_0+1} \left\lbrace \sum_{k=1}^{p} \left( d_{kj} + d_{k(j-1)}  \right) \right\rbrace + o(1).
\end{equation}
\end{theorem}

Theorem~\ref{thm_pred_error} is proved in Appendix~B. Note that this theorem imposes an upper bound on the model sparsity, as the right hand side of \eqref{eq:predconsistency} must go to zero as $ n \rightarrow \infty $. In Section~\ref{sec:consistency}, we specify the limit on the sparsity needed for consistent identification of structural break points.

We now turn to the original problem of estimating the number of break points and  locating them. The next result shows that the number of selected change points, $ \widehat{m} $, based on \eqref{eq_estimation} will be at least as large as the true number, $ m_0 $. Moreover, there exists at least one estimated change point in $ n \gamma_n$-radius neighborhood of each each true change point. 
Before stating the theorem, we need some additional notation. Let $ \mathcal{A}_n = \lbrace t_1, t_2, \ldots, t_{m_0} \rbrace $ be the set of true change points. Following  \cite{Boysen_2009} and \cite{Chan_2014}, define the Hausdorff distance between {
two countable sets in real line} as  
$$ 
d_H (A, B) = \max_{b \in B} \min_{a \in A} |b - a|. 
$$
{
Note that the above definition is not symmetric and therefore not a real distance. However, this is the version of function $ d_H (A, B) $ used in our  next theorem.} 

\begin{theorem}\label{thm_Hausdorff}
Suppose A1--A3 hold. {
Choose $ \lambda_{1,n} =  2 C_1 \sqrt{\frac{\log(n) + 2\log(p) + \log(q)}{n}} $, and $ \lambda_{2,n} = \frac{C_2}{n} \sqrt{\frac{\log p}{n \gamma_n}} $ for some large constants $ C_1 , C_2 > 0 $.} Then, as $ n \rightarrow + \infty $,
$$ 
\mathbb{P} \left(  | \widehat{\mathcal{A}}_n | \geq m_0  \right) \rightarrow 1,
$$
and
$$  
\mathbb{P} \left(  d_H \left( \widehat{\mathcal{A}}_n, \mathcal{A}_n  \right)  \leq n \gamma_n \right) \rightarrow 1. 
$$
\end{theorem}

{
For this theorem, $ \lambda_{1,n} $ could be as large as $ \lambda_{1,n} = O \left( \sqrt{\frac{\gamma_n \log p }{n}} \, \right) $. However, for compatibility, we use the same rate as in Theorem~\ref{thm_pred_error}.} The rate of consistency for break point detection in Theorem~\ref{thm_Hausdorff} is $ n \gamma_n $, which can be chosen as small as possible assuming that Assumptions~A2 and A3 hold. 
$ \gamma_n $ also depends on the minimum distance between consecutive true break points, as well as the number of time series, $ p $. When $ m_0 $ is finite, one can choose $ \gamma_n = {(\log n \log p)}/{n} $ or $ \gamma_n = {(\log\log n \log p)}/{n} $. This means that the convergence rate for estimating the relative locations of the break points, i.e., $ t_j/T $ using $ \widehat{t}_j/T $, could be as low as $ {(\log\log n \log p)}/{n} $. In Section~\ref{sec:comparison}, we compare these rates with those obtained in related procedures.

{
\begin{remark}\label{remark:0}
When $ m_0 $ is known, similar arguments as in the proof of Theorem~\ref{thm_Hausdorff} lead to the same consistency rate, i.e., $ n \gamma_n $, for locating the break points. Later, in Theorem~\ref{thm_selection}, we consider the case of unknown $ m_0 $ and show that the consistency rate for this general case is of order $ B m_0 n \gamma_n {d_n^\star}^2 $ for a large enough constant $ B > 0 $. Compared to the known $ m_0 $ case, the additional term $ B m_0 {d_n^\star}^2 $ quantifies the additional complexity of estimating the number of break points.
\end{remark}
}

The second part of Theorem~\ref{thm_Hausdorff} shows that even though we  select more points than needed, there exists a subset of the estimated points, with cardinality $ m_0 $ that and estimates the true break points at the same rate as if $ m_0 $ was known. This result motivates the second stage of our estimation procedure, discussed in the next section, which removes the redundant break points. 

\section{Consistent Estimation of Structural Breaks}\label{sec:consistency}

Theorem~\ref{thm_Hausdorff} shows that the penalized estimation procedure \eqref{eq_estimation} over-estimates the number of break points. A second stage screening is thus needed to remove the redundant estimated change points and consistently estimate the true change points. To this end, we propose a screening procedure, based on a modification of the procedure in  \citet{Chan_2014}. The main idea is to develop an \emph{information criterion} based on a new penalized least squares estimation procedure, in order to screen the candidate break points found in the first estimation stage. Formally, for a fixed $ m $ and estimated change points $ s_1, \ldots, s_m $, we form the following linear regression:
\begin{equation}\label{eq_regression_second}
\begin{pmatrix} y_q^\prime \\  y_{q+1}^\prime \\ \vdots  \\ y_T^\prime \end{pmatrix} = \begin{pmatrix} Y_{q-1}^\prime  \\ \vdots & 0 & \ldots & 0 \\ Y_{s_1-1}^\prime \\ & Y_{s_1}^\prime \\ 0 & \vdots & \ldots & 0 \\ & Y_{s_2-1}^\prime \\ &&& \\ \vdots & \vdots & \ddots & \vdots \\ &&& \\  &&& Y_{s_m}^\prime \\ 0 & 0 && \vdots \\ &&& Y_{T-1}^\prime \end{pmatrix} \begin{pmatrix} \theta_{(q,s_1)} ^\prime \\  \theta_{(s_{1},s_2)} ^\prime \\ \vdots  \\ \theta_{(s_{m},T)} ^\prime \end{pmatrix} + \begin{pmatrix} \xi_q^\prime \\  \xi_{q+1}^\prime \\ \vdots  \\  \xi_T^\prime \end{pmatrix}.
\end{equation}
$$ 
\mathcal{Y} = \mathcal{X}_{s_1, \ldots, s_m} \theta_{s_1, \ldots, s_m} + \Xi, 
$$
where $ \mathcal{X}_{s_1, \ldots, s_m} \in \mathbb{R}^{n \times \pi_m} $,  $ \theta_{s_1, \ldots, s_m} = \left( \theta_{(q,s_1)}^\prime, \theta_{(s_1,s_2)}^\prime, \ldots, \theta_{(s_m,T)}^\prime \right)^\prime \in \mathbb{R}^{\pi_m \times p} $, with $ \pi_m = (m+1)pq $. We estimate $ \theta_{s_1, \ldots, s_m} $ as the solution of a penalized regression, 
\begin{equation}\label{eq_estimation_second}
\widehat{\theta}_{s_1, \ldots, s_m} = \mbox{argmin}_{\theta_{s_1, \ldots, s_m}} \sum_{i=1}^{m+1} \left(  \frac{1}{s_{i}-s_{i-1}} \sum_{t=s_{i-1}}^{s_i-1}  \| y_t - \theta_{(s_{i-1},s_i)} Y_{t-1} \|_2^2 + \eta_{(s_{i-1},s_{i})} \|\theta_{(s_{i-1},s_i)} \|_1 \right),
\end{equation}
with tuning parameters $ \eta_n = \left( \eta_{(s_{0},s_{1})}, \ldots, \eta_{(s_{m},s_{m+1})} \right) $, where $ s_0 = q $ and $ s_{m+1} = T $. 

Now, let
\begin{equation}\label{ddd}
L_n(s_1, s_2, \ldots, s_m; \eta_n) =  \| \mathcal{Y} - \mathcal{X}_{s_1, \ldots, s_m} \widehat{\theta}_{s_1, \ldots, s_m} \|_F^2 + \sum_{i=1}^{m+1} \eta_{(s_{i-1},s_{i})} \| \widehat{\theta}_{(s_{i-1},s_i)} \|_1,
\end{equation}
{
where $ \| . \|_F $ is the Frobenius norm of the matrix.}
Then, for a suitably chosen sequence $ \omega_n $, specified in Assumption~A4 below, consider the following information criterion:

\begin{equation}\label{eq:IC}
\mathrm{IC} (s_1, \ldots, s_m; \eta_n) =  L_n(s_1, s_2, \ldots, s_m; \eta_n)  + m \omega_n. 
\end{equation}

The second stage of our procedure selects a subset of initial $\widehat{m}$ break points from \eqref{eq_estimation} by solving 
\begin{equation}\label{eq_selection}
( \widetilde{m}, \widetilde{t}_j; j = 1,\ldots,\widetilde{m} ) =  \mbox{argmin}_{0 \leq m \leq \widehat{m}, \, \textbf{s} = (s_1, \ldots, s_m) \in \widehat{\mathcal{A}}_n } \mathrm{IC}(\textbf{s}; \eta_n).
\end{equation}

To establish the consistency of the screening procedure \eqref{eq_selection}, we need two additional assumptions. As before, $ d_n^\star = \sum_{j=1}^{m_0+1} \sum_{k=1}^{p} d_{kj} $ is the total sparsity of the model. 
\begin{itemize}
\item[A4] Let $ \Delta_n = \min_{1 \leq j \leq m_0} |t_{j+1} - t_j| $. Then, $ m_0 n \gamma_n {d_n^\star}^2 /\omega_n \rightarrow 0 $, and $ \Delta_n / (m_0 \omega_n) \rightarrow + \infty $.
\item[A5] There exist a large positive constant $ c > 0 $ such that (a) if $ |s_{i} - s_{i-1}| \leq n \gamma_n $, then $ \eta_{(s_{i-1},s_{i})} = c \sqrt{n \gamma_n \log p} $; (b) if there exist $ t_j$ and $t_{j+1} $ such that $ | s_{i-1} - t_j | \leq n \gamma_n $ and $ | s_i - t_{j+1} | \leq n \gamma_n $, then, $ \eta_{(s_{i-1},s_{i})} = 2 \left( c \sqrt{\frac{\log p}{s_i - s_{i-1}}}   + M_\Phi d_n^\star \frac{n \gamma_n}{s_i - s_{i-1}} \right)  $; (c) otherwise, $ \eta_{(s_{i-1},s_{i})} = 2 \left( c \sqrt{\frac{\log p}{s_i - s_{i-1}}}   + M_\Phi d_n^\star \right)  $.
\end{itemize}

We can now state our main result on consistency of change point detection. 
\begin{theorem}\label{thm_selection}
Suppose A1--A5 hold. Then, as $ n \rightarrow + \infty $, the minimizer  $ ( \widetilde{m}, \widetilde{t}_j, j=1,\ldots,\widetilde{m} ) $ of \eqref{eq_selection} satisfies 
$$ 
\mathbb{P} \left( \widetilde{m} = m_0  \right) \rightarrow 1.
$$
Moreover, there exists a positive constant $ B > 0  $ such that 
$$
\mathbb{P} \left( \max_{1 \leq j \leq m_0} | \widetilde{t}_j - t_j | \leq B m_0 n \gamma_n {d_n^\star}^2 \right) \rightarrow 1.  
$$
\end{theorem}
The proof of Theorem~\ref{thm_selection}, given in Appendix~B, relies heavily on the result presented in Lemma~\ref{lemma_selection}, which is stated and derived in Appendix~A.  
{
While the statement of Lemma~\ref{lemma_selection} is similar to Lemma~6.4 in \citet{Chan_2014}, its proof has important differences to address the additional challenges that arise in our setting. The first difference is due to our definition of information criterion (IC). The IC in \eqref{eq:IC} includes three parts: the sum of squared errors, the $ L_1 $ norm of the estimated VAR parameters in \eqref{eq_estimation_second}, and the penalty on the number of break points, $ \omega_n $. In contrast, the IC in Equation~2.9 of \citet{Chan_2014} does not include an $ L_1 $ penalty. The reason for including this additional penalty is the high-dimensional nature of our problem. Specifically, when the distance between two consecutive estimated break points in the first step is less than $ n \gamma_n $, the restricted eigenvalue condition cannot be verified for the segment defined by these two break points. Therefore, there is no control on the behavior of the estimated VAR parameter within this segment. This makes is harder to evaluate the effect of wrong estimation in the sum of squared term. The addition of the $L_1$ penalty helps circumvent this problem. The second key difference is that the proof of Lemma~6.4 in \citet{Chan_2014} relies on the fact that for each candidate segment, the least squared estimator of the AR parameter has a closed-form representation, which can be utilized to quantify the effect of missing a true break point on the sum of squared error. In contrast, in our high-dimensional setting, there is no closed form solution for \eqref{eq_estimation_second}. This complicates the proof. Thus, in proof of Lemma~\ref{lemma_selection}, each candidate segment in the first step may need a different treatment; also, for some segments, specific restricted eigenvalue and deviation bounds need to be verified in order to characterize the asymptotic behavior of VAR parameters. 
}

\begin{remark}\label{remark:1}
For the case when $ m_0 $ is finite, we can set $ \gamma_n = {(\log n \log p)}/{n} $, $ \lambda_n = o\left( {(\log n \log p)}/{n p} \right) $, and $ \omega_n = {(\log n \log p)}^{1+\nu} $ for some positive $ \nu > 0 $. For these rates, the model can have total sparsity $ d_n^\star = o\left( {(\log n \log p)}^{\nu/2} \right) $. 
\end{remark}

\begin{remark}\label{remark:2}
The proposed procedure can be also applied to low-dimensional time series. For example, with $ p = c n^a $ for positive constants $ c$ and $a$, the probability bounds derived in Lemma~\ref{lemma_bound} would be sharp enough to achieve the desired consistency results {
similar to those for the high-dimensional case in Theorem~\ref{thm_selection}}.
\end{remark}

\begin{remark}\label{remark:3}
Selecting the tuning parameter $ \eta $ in Assumption A5 is challenging in practice, since the distance between candidate break points from the initial estimation to the true break points is unknown. Thus, while the specified tuning parameters achieve optimal consistency rates for locating the break points, they are not practical in finite sample settings and applications. To overcome this challenge, we can instead consider a fixed tuning parameter $ \eta $ in all the candidate segments as $ \eta = C m_0 \sqrt{\Delta_n^\star \log p}/n $ for some large enough positive constant $ C > 0 $, where $ \Delta_n^\star = \max_{1 \leq j \leq m_0} |t_{j+1} - t_j| $. We can still show the consistency of the proposed procedure in \eqref{eq_selection} with this fixed $ \eta $. However, the consistency rate for locating the break points using this fixed rate would be different from that achieved in Theorem~\ref{thm_selection}. For finite $ m_0 $, the rate would be of order $ (n \log p)^{1/2 + \nu} $ for some positive $ \nu > 0 $ as compared to the rate $ {(\log n \log p)}^{1+\nu} $ when the tuning parameters are selected as in Assumption~A5. In all simulation studies and real data applications, $ \eta $ was selected according to the fixed rate mentioned in this remark.
\end{remark}

When the number of change points selected at the first stage $\widehat{m} = | \widehat{\mathcal{A}}_n | $ is large, the second screening step for finding the minimizer of the information criterion $ \mathrm{IC} $ could be computationally demanding. In order to reduce the computational cost, we can use a backward elimination algorithm (BEA), similar to \citet{Chan_2014}, to approximate the optimal point at a lower computational cost. The idea of this algorithm is to start with the full set of selected points, $\widehat{\mathcal{A}}_n$; in each step, we then remove one unnecessary point until no further reduction is possible. Details of the algorithm are given in Appendix~C. 

While the proposed BEA algorithm is not guaranteed to achieve the minimizer of \eqref{eq_selection}, it only requires to search $ \widehat{m}^2 $ sets in order to find the break points. The algorithm thus results in a significant reduction in the computational time when $ \widehat{m} $ is large. The proposed BEA algorithm was used in simulation studies of Section~\ref{sec:sims} and real data examples of Section~\ref{sec:data}, and seems to perform very well in all cases.

\subsection{Comparison with Other Methods}\label{sec:comparison}

Before discussing the parameter estimation consistency in Section~\ref{sec:estimation}, in this section, we compare our rates of consistency for break point estimation with related methods. That is because existing approaches do not provide consistent parameter estimation. By Theorem~\ref{thm_selection}, our rate is of order $ m_0 n \gamma_n {d_n^\star}^2$. Thus, the exact rate depends on values of $ m_0 $, $ \gamma_n $, and $ {d_n^\star} $, which are governed by Assumptions A1--A5.  

\citet{Chan_2014} consider structural break detection in univariate AR processes. When $ m_0 $ is finite, their rate of consistency for estimating the location of break points is of order $  {\log n } $ under the assumption that the distance between two consecutive break points is at least $ {(\log n )}^{1+\nu} $ for some positive $ \nu > 0 $. Our method can be seen as an extension of \citeauthor{Chan_2014} to high dimensions: choosing $ \gamma_n = \log n \log p / n $, the consistency rate of our method is of order $ {\log n \log p} \, {d_n^\star}^2 $ assuming that consecutive break point distances are at least of the same order. The additional factor $ {\log p} \, {d_n^\star}^2 $ quantifies the complexity of the problem in the high-dimensional setting. In the univariate case, or even in multivariate case with fixed dimension $p$, our method achieves the same rates of consistency as \citeauthor{Chan_2014}, i.e., $  {\log n } $. 
 
The test procedure of \citet{aue2009break} can locate break points in the covariance structure of multivariate time series, and achieve a rate of consistency of order $  {\log \log n } $. Interestingly, we can also set $ \gamma_n = \log \log n \log p / n $. In this case, for finite number of break points, our rates for consistently  locating the break points can be $  \log \log n \log p \, {d_n^\star}^2 $, which differs with the rate derived in \citet{aue2009break} by a factor of order $ \log p \, {d_n^\star}^2 $. This factor is, again, due to the fact in our analysis the number of time series, $ p $, is allowed to grow exponentially with $ n $. In contrast, the analysis in \citet{aue2009break} does not directly take the rate of increase in $p$ into account. For fixed $p$, both methods give similar rates for locating the break points, i.e., $  {\log \log n } $.

The recent proposal of \citet{cho2015multiple} uses a CUSUM statistic to identify the number of break points together with their locations. The proposal of \citet{cho2015multiple} is the closest in spirit to our approach as it also identifies structural breaks in high-dimensional time series. However, aside from consistent estimation of model parameters discussed in the Introduction, the two methods have a number of differences. First, in \citet{cho2015multiple}, the number of time series $ p $ is allowed to grow polynomially with the number of time points $ T $. In contrast, we allow $ p $ to grow exponentially with $ T $. Second, the minimum distance between two consecutive break points allowed in \citet{cho2015multiple} is of order $ \Delta_n = T^{\psi} $ for some $ \psi \in (6/7,1) $. In our setting, depending on the sparsity of the model, this rate could be as low as $ \Delta_n = {(\log n \log p)}^{1+\nu} {d_n^\star}^2 $ for some positive $ \nu > 0 $. Therefore, for sparse enough VARs, our method can detect considerably closer break points. Finally, our rate of consistency for estimating the break point locations is of order $ m_0 n \gamma_n {d_n^\star}^2 $, which could be as low as $ {(\log n \log p)}^{1+\nu} $ if we set $ \gamma_n = \log n \log p / n $ and $ d_n^\star = O \big({ ( \log n \log p)}^{\nu/2} \big)$. \citet{cho2015multiple} can achieve a similar rate when $ \Delta_n $ is of order $ T $. However, when $ \Delta_n $ is smaller and is of order $ T^{\psi} $ for some $ \psi \in (6/7,1) $, \citeauthor{cho2015multiple}'s rate of consistency will be of order $ T^{2 - 2 \psi} $, which is larger than our logarithmic rate.

\section{Consistent Parameter Estimation}\label{sec:estimation}
Theorems~\ref{thm_Hausdorff} and \ref{thm_selection} suggest that we can consistently estimate the location (and number) of change points in high-dimensional time series. However, even with such estimates, consistent estimation of $ \Phi^{(.,j)}$ parameters in nonstationary high-dimensional VAR models remains challenging. This challenge primarily stems from the inexact nature of structural break estimation. More specifically, while we know that the estimated break points are in some neighborhood of the true break points, they are not guaranteed to segment the time series into stationary components. A more careful analysis is thus needed to ensure the consistency of VAR parameters. 

The key to our approach for consistent parameter estimation is that Theorems~\ref{thm_Hausdorff} and \ref{thm_selection} imply that removing the selected break points together with large enough $ R_n $-radius neighborhoods will also remove the true break points. We can thus obtain stationary segments at the cost of discarding some portions of the observed time series. Theorem~\ref{thm_Hausdorff} suggests that the radius $R_n$ can be as small as $ n \gamma_n $. However, based on Theorem~\ref{thm_selection}, in order not to keep any redundant break points, $ R_n $ needs to be at least $ B m_0 n \gamma_n {d_n^\star}^2 $ for a large value $ B > 0 $. 

Given the results in Theorems~\ref{thm_selection}, suppose, without loss of generality, that we have selected $ m_0 $ break points using the procedure developed in Section~\ref{sec:consistency}. Denote these estimated break points  by $ \tilde{t}_1, \ldots, \tilde{t}_{m_0} $. Then, by Theorem~\ref{thm_selection}, 
$$
\mathbb{P} \left( \max_{1 \leq j \leq m_0} | \tilde{t}_j - t_j | \leq R_n \right) \rightarrow 1,  
$$
as $ n \rightarrow \infty $. 
Denote $ r_{j1} = \tilde{t}_j - R_n - 1 $, $ r_{j2} = \tilde{t}_j + R_n + 1 $ for $ j = 1, \ldots, m_0 $, and set $ r_{02} = q $ and $ r_{(m_0+1)1} = T $. Further, define the intervals $ I_{j+1} = [ r_{j2}, r_{(j+1)1}  ] $ for $ j = 0, \ldots, m_0 $. The idea is to form a linear regression on $ \cup_{j=0}^{m_0} I_{j+1} $ and estimate the auto-regressive parameters by minimizing an $\ell_1$-regularized least squares criterion. More specifically, we form the following linear regression:
\begin{equation}\label{eq_regression_third}
\begin{pmatrix} y_q^\prime \\ \vdots  \\  y_{r_{11}}^\prime \\  y_{r_{12}}^\prime \\ \vdots \\ y_{r_{21}}^\prime \\ \\ \vdots \\ \\ y_{r_{m_0 1}}^\prime \\ \vdots \\ y_T^\prime \end{pmatrix} = \begin{pmatrix} Y_{q-1}^\prime  \\ \vdots & 0 & \ldots & 0 \\ Y_{r_{11}-1}^\prime \\ & Y_{r_{12}-1}^\prime \\ 0 & \vdots & \ldots & 0 \\ & Y_{r_{21}-1}^\prime \\ &&& \\ \vdots & \vdots & \ddots & \vdots \\ &&& \\  &&& Y_{r_{m_0 2}-1}^\prime \\ 0 & 0 && \vdots \\ &&& Y_{T-1}^\prime \end{pmatrix} \begin{pmatrix} \beta_1^\prime \\  \beta_{2}^\prime \\ \vdots  \\ \beta_{m_0+1}^\prime \end{pmatrix} + \begin{pmatrix}  \zeta_q^\prime \\ \vdots  \\  \zeta_{r_{11}}^\prime \\  \zeta_{r_{12}}^\prime \\ \vdots \\ \zeta_{r_{21}}^\prime \\ \\ \vdots \\ \\ \zeta_{r_{m_0 1}}^\prime \\ \vdots \\ \zeta_T^\prime \end{pmatrix}.
\end{equation}
This regression can be written in compact form as 
$$ 
\mathcal{Y}_{\textbf{r}} = \mathcal{X}_{\textbf{r}} B + E_{\textbf{r}} 
$$
or, in a vector form, as 
\begin{equation}\label{eq_estimation_third_matrix_form}
{\textbf{Y}}_{\textbf{r}} =  \textbf{Z}_{\textbf{r}} \textbf{B}  + \textbf{E}_{\textbf{r}}
\end{equation}
where $ \textbf{Y}_{\textbf{r}} = \mbox{vec}(\mathcal{Y}_{\textbf{r}}) $, $ {\textbf{Z}_{\textbf{r}}} = I_p \otimes \mathcal{X}_{\textbf{r}} $, $ \textbf{B} = \mbox{vec}(B) $, $ {\textbf{E}}_{\textbf{r}} = \mbox{vec}(E_{\textbf{r}}) $, and $ {\textbf{r}} $ is the collection of all $ r_{j1}$ and $ r_{j2} $ for $ j = 0, \ldots, m_0+1 $. Let $ \tilde{\pi} = (m_0+1) p^2 q $, $ N_j = \mbox{length} (I_{j+1}) = r_{(j+1)1}-r_{j2} $ for $ j = 0, \ldots, m_0 $ and $ N = \sum_{j=1}^{m_0} N_j $.
Then, $ \textbf{Y}_{\textbf{r}} \in \mathbb{R}^{N p \times 1} $, $ {\textbf{Z}_{\textbf{r}}} \in \mathbb{R}^{N p \times \tilde{\pi}} $, $ \textbf{B} \in \mathbb{R}^{\tilde{\pi} \times 1} $, and $ {\textbf{E}}_{\textbf{r}} \in \mathbb{R}^{N p \times 1} $. We estimate the VAR parameters by solving
\begin{equation}\label{eq_estimation_third}
\widehat{\textbf{B}} = \mbox{argmin}_{\textbf{B}} {N}^{-1} \left\| \textbf{Y}_{\textbf{r}} - {\textbf{Z}_{\textbf{r}}} \textbf{B} \right\|_2^2 + \rho_n  \left\| \textbf{B} \right\|_1.
\end{equation}

We obtain the following consistency result. 
\begin{theorem}\label{thm_parameter_consistency}
Suppose A1--A5 hold and $ m_0 $ is unknown and $ R_n = B m_0 n \gamma_n {d_n^\star}^2 $. Assume also that $ \Delta_n > \varepsilon n $ for some large positive $ \varepsilon > 0 $ and $ \rho_n = C \sqrt{\frac{\log \tilde{\pi}}{N}} $ for large enough $ C > 0 $. (Note that $ N/n = O(1) $.) Then, as $ n \rightarrow + \infty $, the minimizer  $ \widehat{\textbf{B}} $ of \eqref{eq_estimation_third} satisfies 
$$ 
\left\| \widehat{\mathbf{B}} - \Phi  \right\|_\ell = O_p \left( (d_n^\star)^{1/\ell} \rho_n \right) \mbox{ for } \ell = 1, 2.
$$
\end{theorem}
Theorem~\ref{thm_parameter_consistency} is proved in Appendix B, where it is also shown that, if $ m_0 $ is known, then it is enough to set $ R_n = n \gamma_n $. 

A natural question is whether consistent parameter estimation can be achieved without the third step of our procedure. The next remark emphasizes that removing the $R_n$-radius of estimated break points and fast rates of convergence play a crucial role in consistent parameter estimation.

{
\begin{remark}
Existing approaches for change point detection in multivariate and high-dimensional settings, including ours and the proposal of  \citet{cho2015multiple}, are only able to consistently estimate the \emph{relative location} of change points, $ t_j / T $. As a result, using the estimated change points to partition the time series into `stationary' segments, and then estimating the VAR parameters for each segment using a regularized estimation procedure may not lead to consistent estimates of VAR parameters. This is because if the distance between the break point $ t_j $ and its estimation $ \tilde{t}_j $ is large relative to the minimum distance between two consecutive break points, then the restricted eigenvalue (RE) condition may not hold for the specific segment starting with $ t_j $. As discussed in \citet{BickelETAL_2009}, the RE condition is critical for establishing the consistency of parameter estimation in high-dimensional settings. Hence, removing the $R_n$-radius of estimated break points, together with our improved rates of change point detection, are key in achieving parameter estimation consistency in high-dimensional VAR models. 
In contrast, the consistency rate for locating break points using the method of \citet{cho2015multiple} is of order $ T^{2-2\psi} $ with $ \psi \in (6/7,1) $. Now, if, similar to our setting, the minimum distance allowed between two consecutive break points in the model is of order $ \left( \log T  \, \log p \right)^{1+\nu} $ for some $ \nu > 0 $, then the estimated break points might be far from the true break points since $ {T^{2-2\psi}}/{\left( \log T  \, \log p \right)^{1+\nu}} $ may diverge. This violates the RE condition and may result in inconsistent parameter estimation using $L_1$-reqularization. 
\end{remark}
}


{
\section{Tuning Parameter Selection}\label{sec:tuningselection}

Our proposed three-stage procedure relies on multiple tuning parameters, $ \lambda_{1,n}, \lambda_{2,n}, \eta_n, \omega_n $, $ \rho_n $ and $R_n$. Although the theoretical rates for these parameters are derived in previous sections, their selection in finite sample applications needs further discussion. In this section, we provide guidance on selecting these tuning parameters.

\begin{itemize}
\item[$ \lambda_{1,n} $:] We can select $ \lambda_{1,n} $ by cross-validation. Let $ \mathcal{T} $ be a set of equally spaced time points, starting from a randomly selected initial time point. The data without observations in $ \mathcal{T} $ can then be used in the first step of our procedure to estimate $\Theta$ for a range of values for $ \lambda_{1,n} $. The parameters estimated in the first step are then used to predict the series at time points in $ \mathcal{T}$ . The value of $ \lambda_{1,n} $ which minimizes the mean squared prediction error over $ \mathcal{T} $ is the cross-validated choice of $ \lambda_{1,n} $.

\item[$ \lambda_{2,n} $:] As described previously, the rate for $ \lambda_{2,n} $ vanishes fast as $ T $ increases. Thus for simplicity, we suggest setting $ \lambda_{2,n} $ to zero. This choice was used in all of the numerical analyses in the paper---both simulations studies and real data applications---and seems to give satisfactory results. 

\item[$ \eta_n $:] Following Remark~\ref{remark:3}, we suggest a fixed rate for 
this parameter, and, based on the theoretical rate in Remark~\ref{remark:3}, set $ \eta_n = {(\log n \log p)}/{n} $. This choice was used in all of the numerical analyses in the paper and provides very good results.

\item[$ \omega_n $:] The selection of $ \omega_n $ is the most difficult,  since it depends on how large changes in the VAR parameters must be in order to consider them as break points in finite sample applications. A similar tuning parameter  appears in other work on the topic, including \citet{cho2015multiple} and \citet{Chan_2014}. Following Assumption~A4, in our analysis we set $ \omega_n = C {\left( \log n \log p \right)}^{3/2} $ for some $ C > 0 $. The range for the constant $ C $ used in our analysis is within the interval $ [0, 1]$. 

\item[$ \rho_n $:] Finally, we also need to select the tuning parameter $ \rho_n $ for parameter estimation. We select $ \rho_n $ as the minimizer of the combined Bayesian Information Criterion (BIC) for all the segments. Following \citet{lutkepohl2005new} and \citet{zou2007degrees}, for $ j = 0, \ldots, \tilde m $ we define the BIC on the interval $ I_{j+1} = [ r_{j2}, r_{(j+1)1}  ] $ as
$$ 
\mbox{BIC}(j, \rho_n) = \log (\mbox{det} \widehat{\Sigma}_{\varepsilon,j} ) + \frac{\log (r_{(j+1)1} - r_{j2})}{(r_{(j+1)1} - r_{j2})} \left\| \widehat{\beta}_{j+1} \right\|_0,  
$$
where $ \widehat{\Sigma}_{\varepsilon,j} $ is the residual sample covariance matrix with $ \widehat{\textbf{B}} $ estimated  in \eqref{eq_estimation_third}, and $ \| \widehat{\beta}_{j+1} \|_0 $ is the number of nonzero elements in $ \widehat{\beta}_{j+1} $; then $ \rho_n $ is selected as
\begin{equation}\label{eq_BIC}
 \widehat{\rho}_n = \mbox{argmin}_{\rho_n} \sum_{j=0}^{\tilde m} \mbox{BIC}(j, \rho_n). 
\end{equation}
{
\item[$ R_n $:] Recall from Section~\ref{sec:estimation} that we need to remove the selected break points together with their $R_n$-radius neighborhood before estimating the parameters using \eqref{eq_estimation_third}. In practice, the radius $ R_n $ needs to be estimated. However, a closer look at the proof of Theorem~\ref{thm_selection} together with Assumption~A4 suggest that $ \omega_n $ can be chosen as an upper bound for the selection radius $ R_n $. In other words, in the statement of Theorem~\ref{thm_selection}, the radius $ B m_0 n \gamma_n {d_n^\star}^2 $ can be replaced by $ \omega_n $ and the result would still hold. Formally,
$$
\mathbb{P} \left( \max_{1 \leq j \leq m_0} | \widetilde{t}_j - t_j | \leq \omega_n \right) \rightarrow 1, 
$$
as $ n \rightarrow \infty $. Therefore, in all simulation scenarios and data applications, we set $ R_n = \omega_n $.}
\end{itemize}
}

%

\section{Simulation Studies}\label{sec:sims}

\subsection{Preliminaries}\label{sec:simsec1}
We evaluate the performance of the proposed three-stage estimator with respect to both structural break detection and parameter estimation.  {
In this section, we consider three simulations scenarios. The first two scenarios examine low-dimensional settings, with $ T = 300 $, $ p = 20 $, $ q = 1 $ and $ m_0 = 2 $; the third scenario examines a high-dimensional setting with $ T = 80 $, $ p = 100 $, $ q = 1 $ and $ m_0 = 1 $. Detail of simulation settings in each scenario are explained in Section~\ref{sec:simsec2} and Appendix~D, where results from two additional simulation settings are also presented\footnote{See also https://github.com/abolfazlsafikhani/SBDetection.}.
In all, except Scenario~5 in Appendix~D, results are averaged over 100 randomly generated data sets with mean zero and $ \Sigma_\varepsilon = 0.01 I_T $; a non-diagonal $\Sigma_\varepsilon$ is considered in Scenario~5.}

For structural break detection, we compare our method with the sparsified binary segmentation-multivariate time series (SBS-MVTS) approach of \citet{cho2015multiple}. For both methods, we report the locations of the estimated break points and the percentage of simulations where each break point is correctly identified. 
This percentage is calculated as the proportion of simulations, where selected break points are close to each of the true break points. More specifically, a selected break point is counted as a `success' for the first true break point, $t_1$, if it is in the interval $ [ 0, t_1 + 0.5 (t_2 - t_1)) $; similarly, a selected break point is counted as a `success' for the second true break point, $t_2$, if it falls in the interval $ [ t_1 + 0.5 (t_2 - t_1), T ] $. 
For our method, we also report the percentage of cases where the break point is correctly estimated in the $R_n$-radius of truth. 
The results are reported in Table~\ref{table_sim_all}.

For parameter estimation, we evaluate the performance of our procedure by reporting the (relative) estimation error, as well as true and false positive rates, calculated by comparing the nonzero patterns of true and estimated coefficient matrices. {
Since existing procedures do not provide consistent estimates of parameters in piecewise stationary VAR models, we compare our procedure to a procedure based on the state-of-the-art methods in high-dimensional change point detection and VAR estimation. More specifically, we use the estimated change points from SBS-MVTS in order to partition the time series into `stationary' segments, but without removing the $ R_n$-radius of selected break points. We then use our $\ell_1$-penalized procedure \eqref{eq_estimation_third} to estimate the VAR parameters in each segment. The results in Table~\ref{table_sim_1_estimation} clearly show the advantages of our procedure and highlight the importance of removing the $ R_n$-radius of selected break points in the third step of our procedure; see Section~\ref{sec:simsec2} for details.}

\begin{figure}[t!]
\begin{center}
\includegraphics[width=0.45\textwidth, clip=TRUE, trim=0cm 0cm 0cm 1.5cm]{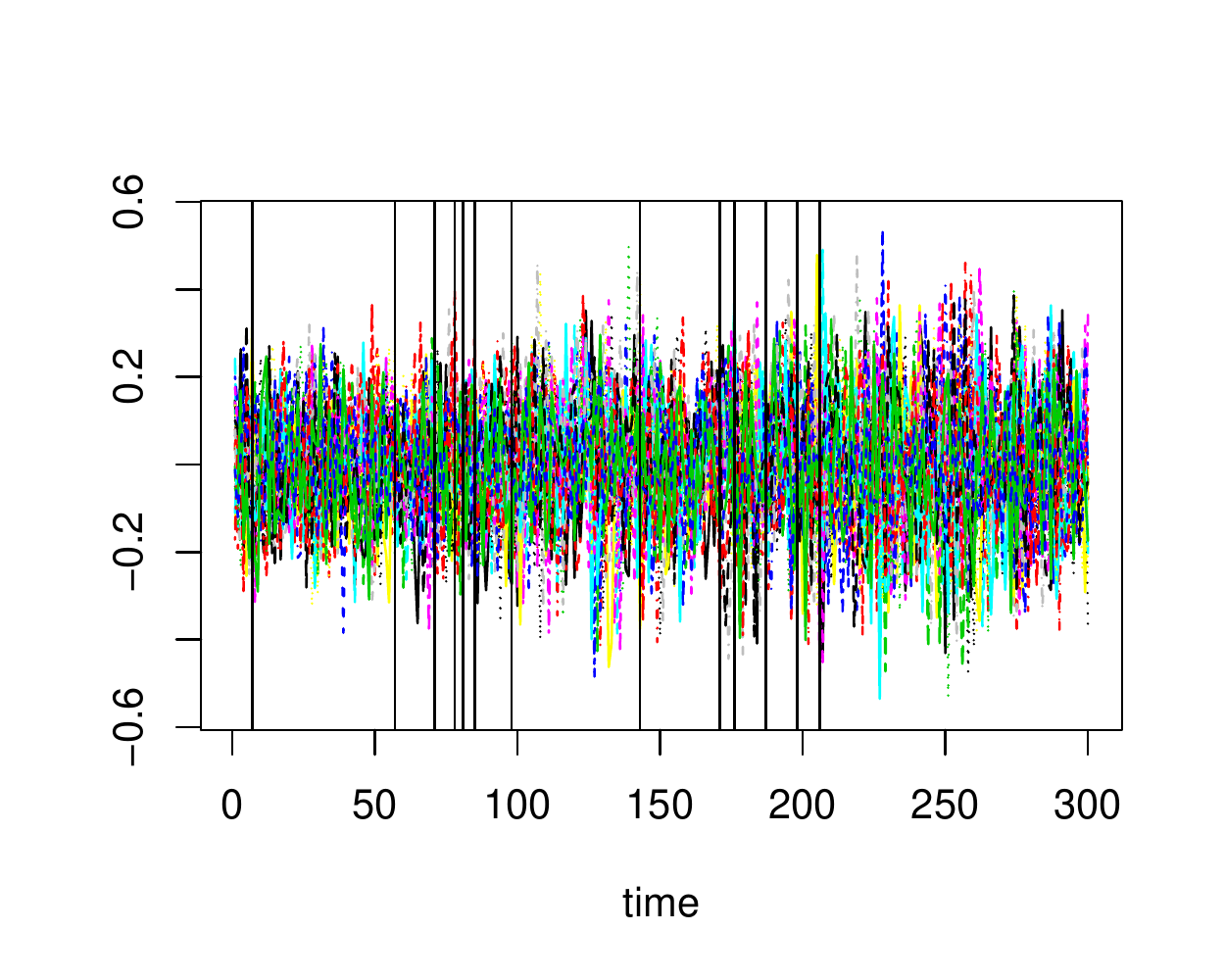}
\includegraphics[width=0.45\textwidth, clip=TRUE, trim=0cm 0cm 0cm 1cm]{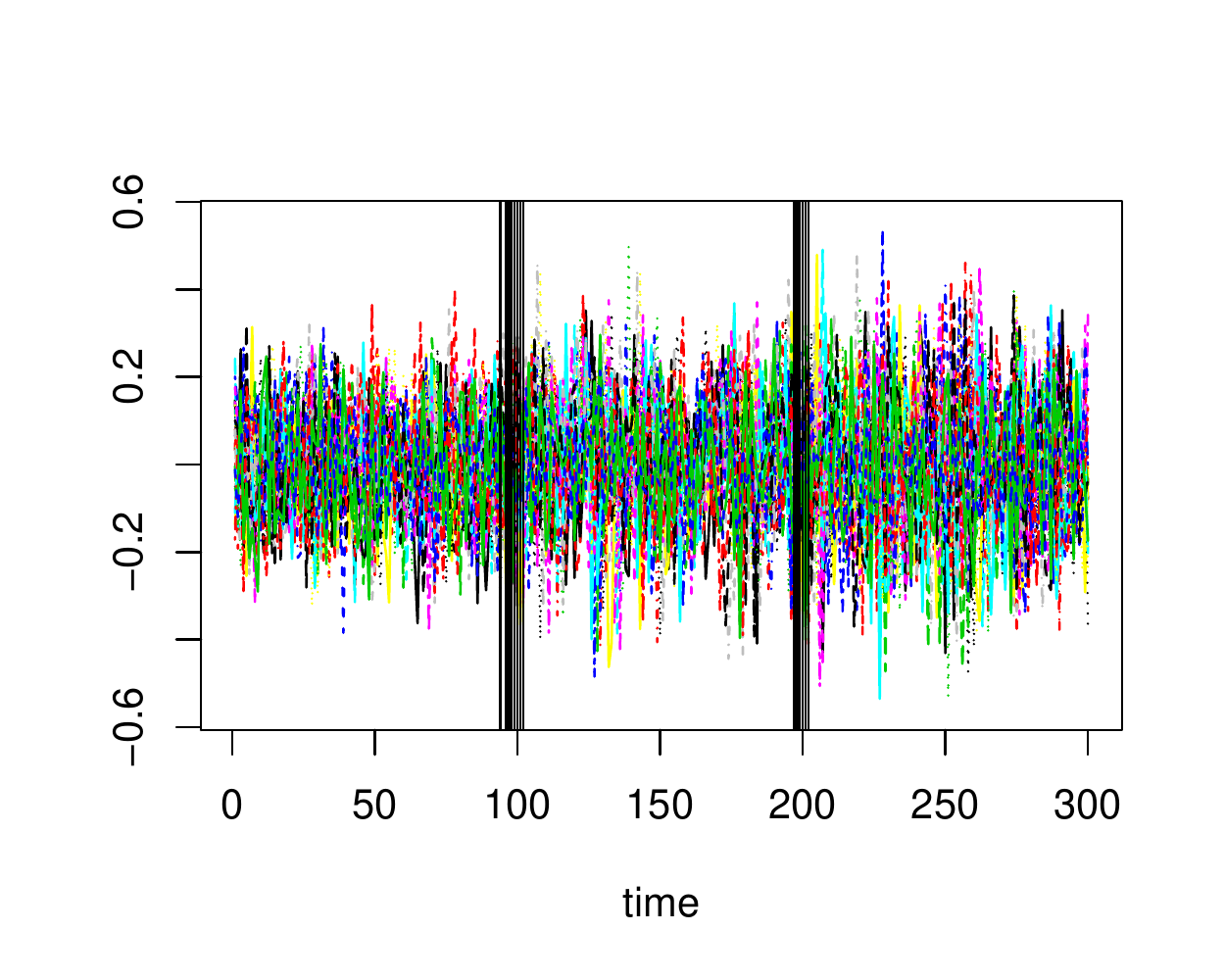}
\vspace{-0.5cm}
\caption{Left: Estimated break points from the first stage of our proposed procedure (Equation~\ref{eq_estimation}) for a single runs in Simulation Scenario 1; on average $\sim$13 points are selected in the first stage. Right: Final selected break points for all 100 simulation runs in Simulation Scenario 1.}
\label{fig_step_1}
\end{center}
\end{figure}

\subsection{Simulation Scenarios}\label{sec:simsec2}

\textit{Simulation Scenario 1 (Simple $ \Phi $ and break points close to the center)}. In the first scenario, the autoregressive coefficients are chosen to have the same structure but different values; 
{see Appendix~D for details.} In this scenario, $ t_1 = 100 $ and $ t_2 = 200 $, which means that break points are not close to the boundaries. 

{
Before comparing our procedure with the SBS-MVTS method of \citet{cho2015multiple}, we take a closer look at the first two steps of our procedure in this simulation setting. 
}
To this end, we plot the break points in one simulated data sets in the left panel of Figure~\ref{fig_step_1}. As expected from Theorem~\ref{thm_Hausdorff}, more than 2 break points are detected using the first stage estimator. However, some break points are indeed in a small neighborhood of true change points. Our second-stage screening procedure eliminates the extra candidate points, leaving only the two closest points to the true change points. The final selected points in all 100 simulation runs are shown in the right panel of Figure~\ref{fig_step_1}, and confirm the consistency of the proposed method in selecting the break points. 

\emph{Simulation Scenario 2 (Simple $ \Phi $ and break points close to the boundaries)}. The coefficient matrices in this simulation are similar to those in Scenario 1. However, in this scenario, the break points are closer to the boundaries. Specifically, $ t_1 = 50 $ and $ t_2 = 250 $. 

{
\textit{Simulation Scenario 3 (High-dimensional setting with simple $ \Phi $)}. In this scenario, $ T = 80 $, $ p = 100 $ and there is a single break point at $ t = 40 $. The autoregressive coefficients are chosen to have the same structure as the first two VAR matrices in Scenario~1.
}

\begin{table}
\hspace{-1cm}
\caption{\label{table_sim_all} Mean and standard deviation of estimated break point locations, the percentage of simulation runs where break points are correctly detected (selection rate)}
\centering
\begin{tabular}{lcccccc} 
  \hline
& method & break point & truth & mean & std & selection rate \\ 
  \hline
    \hline
    Scenario 1 & & & & & \\
    & SBS-MVTS & 1 & 0.3333 & 0.3513 & 0.039 & 0.85  \\ 
    \, & &  2 & 0.6667 & 0.6425 & 0.0558 & 0.87 \\
   & Our method & 1 & 0.3333 & 0.3318 & 0.0104 & 1 \\ 
  \,& & 2 & 0.6667 & 0.6584 & 0.0153 & 1 \\ 
    Scenario 2 & & & & & \\
      & SBS-MVTS & 1 & 0.1667 & 0.31 & 0.0802 & 0.94 \\ 
  \,& & 2 & 0.8333 & 0.6414 & 0.102 & 0.68 \\
  & Our method  & 1 & 0.1667 & 0.1763 & 0.022 & 1 \\ 
  \,& & 2 & 0.8333 & 0.7971 & 0.023 & 1 \\ 
    Scenario 3 & & & & & \\
        & SBS-MVTS & 1 & 0.5 & 0.4688 & 0.1504 & 0.64 \\ 
    & Our method & 1 & 0.5 & 0.4975 & 0.0223 & 1 \\
   \hline
\end{tabular}

\end{table}

\subsection{Simulation Results}\label{sec:simsec3}

The mean and standard deviation of locations of selected break points, relative to the sample size $ T $ --- i.e., $\widetilde{t}_1 / T $ and $ \widetilde{t}_2 / T $ --- for all simulation scenarios are summarized in Table~\ref{table_sim_all}. The results clearly indicate that, in all settings, our procedure accurately detects both the number of break points, as well as their locations. They also suggest that our procedure produces more accurate estimates of break points than the SBS-MVTS method. The advantage of our method is more pronounced when comparing the percentage of times where the break points are correctly detected using each method. 

{
In addition to overall improvement in detecting break points, the simulation results also indicate that our procedure offers significant advantages over SBS-MVTS in in Simulation Scenarios 2 and 3. In Scenario~2, where the break points are closer to the boundaries, the detection performance of our procedure does not deviate much from Scenario 1, but the performance on locating the break points seems to be slightly worse. This is in stark contrast to SBS-MVTS, which gives worse estimates of break points in this simulation setting.
Similarly, our procedure continues to perform well in the high-dimensional setting of Scenario 3, whereas SBS-MVTS performs considerably worse than Scenario 1.}

Table~\ref{table_sim_1_estimation} summarizes the results for autoregressive parameter estimation in all three simulation scenarios. The table shows mean and standard deviation of relative estimation error (REE), as well as true positive (TPR) and false positive rates (FPR) of the estimates. The results suggest that our method also performs well in terms of parameter estimation. However, true positive rates are low in Scenario 2. One potential explanation for the reduced TPR in this scenario is that the first and third segments used for estimation in this scenario contain less than 40 time points,  compared to 90 time points in Scenario~1. This shorter length makes it harder to estimate the parameters and correctly select zero and nonzero coefficients. {
Comparing our procedure and the na\"ive procedure based on SBS-MVTS, introduced in Section~\ref{sec:simsec1}, indicates that our procedure is superior in all simulation scenarios, both in terms of estimation error and variable selection.
Since the two methods use the same estimation procedure, this advantage can be attributed to our proposal, and, in particular, to removing the $R_n$-neighborhood of the estimator break points. These findings affirm our theoretical results.} 

\begin{table}
\hspace{-1cm}
\caption{\label{table_sim_1_estimation} Results of parameter estimation for all three simulation scenarios. The table shows mean and standard deviation of relative estimation error (REE), true positive rate (TPR), and false positive rate (FPR) for estimated coefficients.}
\centering
\begin{tabular}{lccccc} 
  \hline
 & Method & REE & SD(REE) & TPR & FPR \\
  \hline
    \hline 
  Simulation 1   & & & & &  \\
  & Our Method   & 0.3385 & 0.0282 & 1.00 & 0.036  \\
  & SBS-MVTS     & 0.4113 & 0.2678 & 1.00 & 0.039  \\
  Simulation 2   & & & & &  \\
  & Our Method   & 0.654 & 0.052 & 0.72 &  0.03  \\ 
  & SBS-MVTS     & 0.901 & 0.154 & 0.63 & 0.03  \\
    Simulation 3   & & & & &  \\
  & Our Method   & 0.6422 & 0.0234 & 0.91 & 0.003 \\ 
  & SBS-MVTS     & 0.8023 & 0.089 & 0.67 & 0.003  \\ 
   \hline
\end{tabular}
\end{table}

\section{Applications}\label{sec:data}

\subsection{EEG Data}\label{sec:EEG}
In this section, we revisit the EEG data discussed in Section~\ref{sec:intro}. Recall that the data consists of electroencephalogram (EEG) signals recorded at 18 locations on the scalp of a patient diagnosed with left temporal lobe epilepsy during an epileptic seizure. The sampling rate is 100 Hz and the total number of time points per EEG is $ T =22,768 $ over 228 seconds. The time series for all 18 EEG channels are shown in Figure~\ref{fig_EEG_full}. The seizure was estimated to take place at $ t = 85 s $. Examining the EEG plots, it can be seen that the magnitude and the volatility of signals change simultaneously around that time. 
To speed up the computations, we select ten observation per second and reduce the total time points to $ T = 2276 $. 

Data from one of the EEG channels (P3) was previously used by \cite{Davis_2006} and \cite{Chan_2014} for detecting structural breaks in the time series. As a comparison, we apply the SBS-MVTS method of \cite{cho2015multiple} as well as our procedure to detect the break points based on changes in all 18 time series. Table~\ref{table_EEG} shows the location of the selected break points using the Auto-PARM method of \cite{Davis_2006} and the two-stage procedure of \cite{Chan_2014}, based on data from channel P3, as well as those estimated using our method and SBS-MVTS based on all 18 channels. The selected break points by our method are also shown in Figure~\ref{fig_EEG_selected}. 

Our method detects a break point at $ t = 83 $, which is close to the seizure time identified by neurologists. Most other break points selected by our method are also close to break points detected by the two univariate approaches and SBS-MVTS. However, the main advantage of our method is that it also provides consistent estimates of VAR parameters. As shown in Figure~\ref{fig_EEG_network}, these estimates can be used to gain novel insight into changes in mechanisms of neuronal interactions before and after seizure. 

Given the proximity of selected break points between $t=83$ and $t=162$, in order to obtain the networks in Figure~\ref{fig_EEG_network}, we consider the time segments before and after these two time points. More specifically, using the procedure of Section~\ref{sec:estimation}, we discard observations in the $R_n$ radius before $t=83$ and after $t=162$ in order to ensure the stationarity of remaining observations. We then use the $\ell_1$-penalized least square estimator of \eqref{eq_estimation_third}, with tuning parameter selected by BIC \eqref{eq_BIC}, to obtain estimates of VAR parameters before and after seizure. Network edges in Figure~\ref{fig_EEG_network} correspond to nonzero estimated coefficients that are at least larger than 0.05 in magnitude. This thresholding is motivated by the known over-selection property of lasso \citep{ShojaieBasuMichailidis_2012} and is used to improve the interpretability of estimated networks. 

\begin{figure}[t]
\begin{center}
\includegraphics[width=0.5\linewidth, clip=TRUE, trim=0mm 0mm 00mm 20mm]{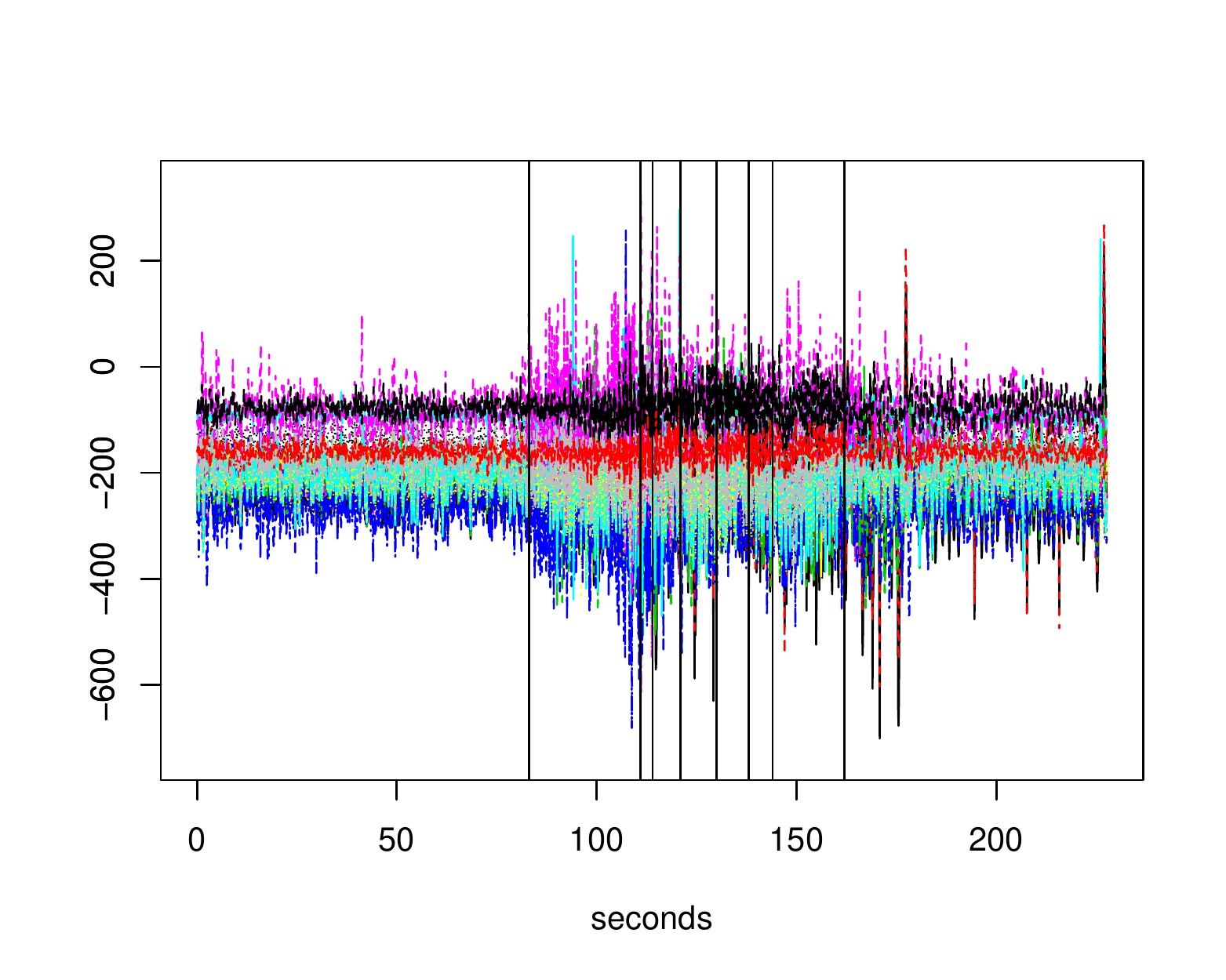}
\vspace{-0.5cm}
\caption{EEG data over 228 seconds with the 8 selected break points using our method.}\label{fig_EEG_selected}
\end{center}
\end{figure}

\begin{table}
\caption{\label{table_EEG} Location of break points detected in the EEG data using four estimation methods. The locations are rounded to the closest integer.}
\centering
\begin{tabular}{llllllllllll}
\hline
Methods & 1 & 2 & 3 & 4 & 5 & 6 & 7 & 8 & 9 & 10 & 11 \\ 
  \hline
  \hline
  Auto--PARM & 86 & 90 & 106 & 121 & 133 & 149 & 162 & 175 & 206 & 208 & 326 \\ 
   Chan et al. (2014) & 84 & 106 & 120 & 134 & 155 & 177 & 206 & 225 & -- & -- & -- \\ 
SBS--MVTS & 84 & 107 & 114 & 126 & 133 & 143 & 157 & 176 & -- & -- & --  \\ 
  Our method & 83 & 111 & 114 & 121 & 130 & 138 & 144 & 162 & -- & -- & --  \\ 
   \hline
\end{tabular}

\end{table}

\subsection{Yellow Cab Demand in NYC}\label{sec:taxi}

As a second example, we apply our method and SBS-MVTS to the yellow cab demand data in New York City (NYC), obtained from the NYC Taxi \& Limousine Commission's website\footnote{http://www.nyc.gov/html/tlc/html/about/trip--record--data.shtml}. Here, the number of yellow cab pickups are aggregated spatially over the zipcodes and temporally over 15 minute intervals during April 16th, 2014. We only consider the zipcodes with more than 50 cab calls to obtain a better approximation using linear VAR models. This results in time series for 39 zipcodes observed over 96 time points. To identify structural break points, we consider a differenced version of the data in order to remove first order non-stationarities.

\begin{table}[h]
\caption{\label{table_Taxi} The location of break points for the NYC Yellow Cab Demand data.}
\centering
\begin{tabular}{llllllllllll}
\hline
  & 1 & 2 & 3 & 4 & 5   \\ 
  \hline
  \hline
SBS--MVTS & 6am & 11:30am & -- & -- & --   \\
  Our method & 7am & 8:15am & 10am &  6pm & 7pm  \\ 
   \hline
\end{tabular}

\end{table}

Table~\ref{table_Taxi} shows the 5 break points detected by our method, along with two break points identified by SBS-MVTS; the differenced time series and  detected break points by our method are also shown in Figure~\ref{fig_Taxi_selected}. 
Based on data from NYC Metro (MTA), morning rush hour traffic in the city occurs between 6:30AM and 9:30AM, whereas the afternoon rush hour starts from 3:30PM and continues until 06:00PM. Interestingly, the selected break points by our method are very close to the rush hour start/end times during a typical day. Specifically, the selected break points at 7AM, 10AM, and 6PM are close to rush hour periods in NYC. These results suggest that the covariance structure of cab demands between the zipcodes in NYC may significantly change before and after the rush hour periods. Even with the least conservative tuning parameters, SBS-MVTS only selects two break points in this example, and does not identify any break points in the afternoon rush hour period.

\begin{figure}[h]
\begin{center}
\includegraphics[width=0.5\linewidth, clip=TRUE, trim=0cm 0cm 0cm 2cm]{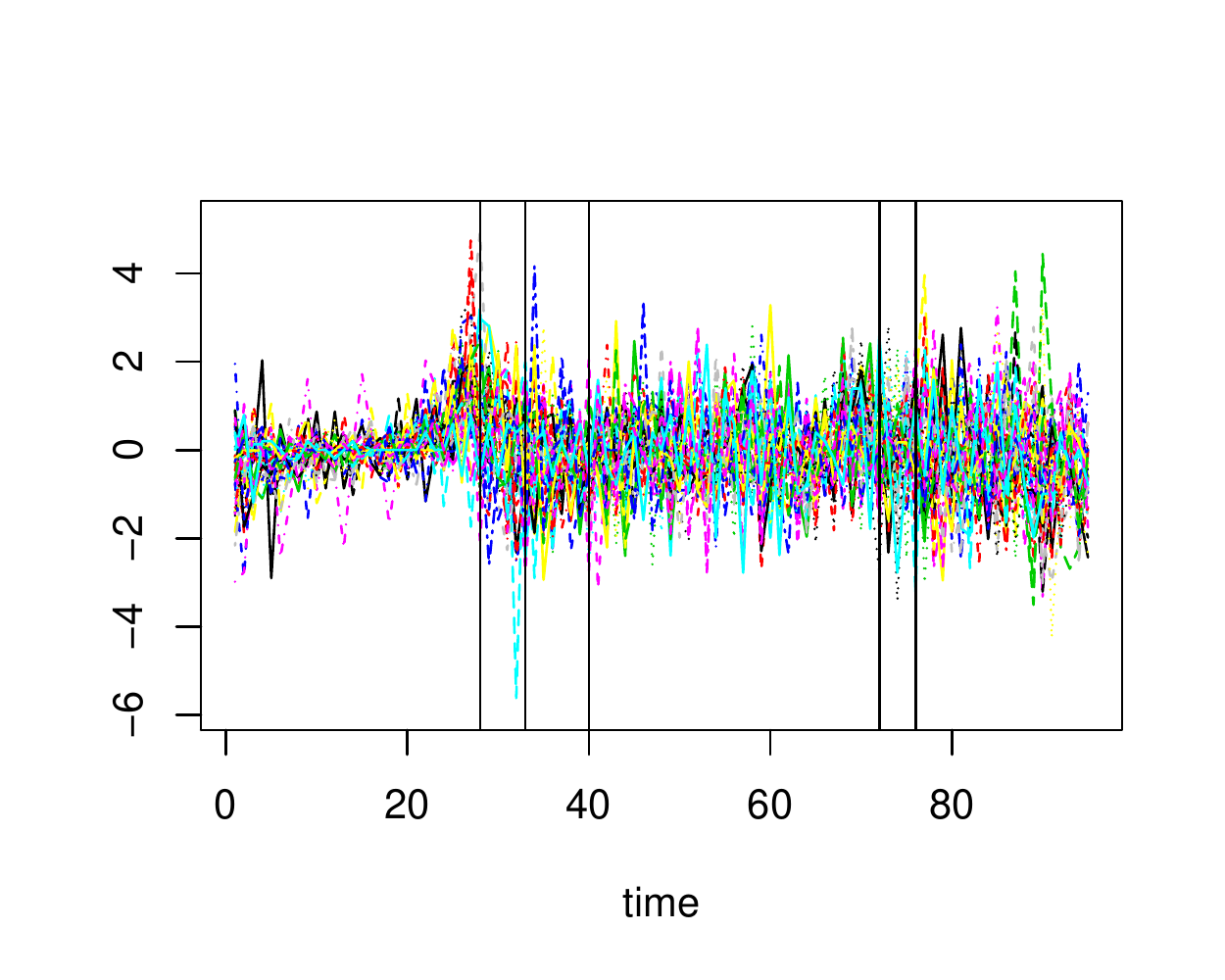}
\vspace{-0.5cm}
\caption{The NYC Yellow Cab Demand differenced time series from 39 different zipcodes over a single day with 96 time points. The 5 selected break points by the proposed method are shown as vertical lines.}\label{fig_Taxi_selected}
\end{center}
\end{figure}

\section{Discussion}\label{sec:disc}

We proposed a three-stage method for simultaneous detection of structural break points and parameter estimation in high-dimensional piecewise stationary VAR models. 

We showed that the proposed method consistently estimates the total number and location of break points. Moreover, it consistently estimates the parameters of the underlying high-dimensional sparse piecewise stationary VAR model.  Numerical experiments in three simulation settings and two real data applications corroborate these theoretical findings. In particular, in both real data examples considered, the break points detected using the proposed method are in agreement with the nature of the data sets. 

When the total number of break points, $ m_0 $, is finite, the rate of consistency for detecting break point locations relative to the sample size $ T $ depends on three factors: (1) the number of time points, $ T $, (2) the number of time series, $ p $, and (3) the total sparsity of the model, $ d_n^\star $. In the univariate case, \cite{Chan_2014} obtained a consistency rate of order $ ({\log n})/{n} $. In the high-dimensional case, the rate shown here is of order $ ({{d_n^\star}^2 \log n \log p })/{n} $. The $\log p$ and $d_n^\star$ factors in this rate highlight the challenges of change point detection in high dimensions. 
The proposed procedure also allows the number of break points to increase with the sample size, as long as the minimum distance between consecutive break points is large enough, as characterized by Assumptions A3 and A4. 

A limitation of the proposed procedure is the need to select multiple tuning parameter. Among these, selecting the penalty parameter for the second stage estimator \eqref{eq_estimation_second} can be challenging in practice. In the numerical studies in this paper, a simplified version of this tuning parameter was used. However, this simplified version does not guarantee optimal rates of consistency for break point estimation. Investigating optimal choices of tuning parameters for the proposed procedure can be a fruitful area of future research. 
Theoretical analysis of data-driven choices of tuning parameters discussed in Section~\ref{sec:tuningselection} can be another future of direction research. 

\bibliography{Safikhani_Reference}

\clearpage
\section*{Appendix}\label{sec:appendix}
In Appendix~A, we collect technical lemmas needed to prove the main results. Proofs of the main results are given Appendix~B. 
{Details of the algorithm for solving the optimization problem \eqref{eq_estimation} are given in Appendix~C. Finally, further details on simulations settings, and additional simulation results are reported in Appendix~D.}

\subsection*{Appendix~A: Technical Lemmas}

\begin{lemma}\label{lemma_first}
There exist constants $ c_i > 0 $ such that for $ n \geq c_0 \left( \log(n) + 2\log(p) + \log(q)  \right) $, with probability at least $ 1 - c_1 \exp \left( - c_2 \left( \log(n) + 2\log(p) + \log(q)  \right) \right) $, we have 
\begin{equation}
\left|\left|  \frac{\textbf{Z}^\prime \textbf{E} }{n}  \right|\right|_\infty \leq c_3 \sqrt{\frac{\log(n) + 2\log(p) + \log(q)}{n}}
\end{equation}
\end{lemma}

\begin{proof}
Note that $ \frac{1}{n} \textbf{Z}^\prime \textbf{E} = \frac{1}{n} (I_p \otimes \mathcal{X}^\prime) \textbf{E} = \mbox{vec} (\mathcal{X}^\prime E)/n $. Let $ \mathcal{X}(h,.) $ and $ \mathcal{X}(h,l) $ be the $ h$-th block column and the $ l$-th column of the $ h$-th block column of $ \mathcal{X} $, respectively, for $ 1 \leq h \leq n $, $ 1 \leq l \leq d $. More specifically, 

\begin{equation}
\mathcal{X}(h,.) = \begin{pmatrix}  & 0 & \\ & \vdots & \\ & 0 & \\  y_{q+h-2}^\prime & \ldots & y_{h-1}^\prime \\ & \vdots & \\  y_{T-1}^\prime & \ldots & y_{T-q}^\prime \end{pmatrix}_{n \times pq},  \hspace{1cm} \mathcal{X}(h,l) = \begin{pmatrix} 0 \\  \vdots \\ 0 \\  y_{q+h-l-1}^\prime \\  \vdots  \\ y_{T-l}^\prime \end{pmatrix}_{n \times p}. 
\end{equation} 
Now, 
\begin{equation}
\left\|  \frac{\textbf{Z}^\prime \textbf{E} }{n}  \right\|_\infty = \max_{1 \leq h \leq n, 1 \leq l \leq d, 1 \leq i,j \leq p} \left|  e_i^\prime \left( \frac{\mathcal{X}^\prime (h,l) E }{n}  \right) e_j  \right|,
\end{equation}
where $ e_i \in \mathbb{R}^p $ with the $ i$-th element equals to 1 and zero on the rest. Note that,
$$  \frac{\mathcal{X}^\prime (h,l) E }{n} = \frac{1}{n} \sum_{t=h-l-1}^{T-q-l} y_{q+t} \varepsilon^\prime_{q+t+l}. $$

Now, since $ \mbox{cov} (y_{q+t},  \varepsilon_{q+t+l}) = 0 $ for all $ t, l, h $, an  argument similar to Proposition~2.4(b) of \cite{Basu_2015} shows that for fixed $ i, j, h, l $, there exist $ k_1, k_2 > 0 $ such that for all $ \eta > 0 $:
$$ 
\mathbb{P} \left( \left|  e_i^\prime \left( \frac{\mathcal{X}^\prime (h,l) E }{n}  \right) e_j  \right| > k_1 \eta  \right) \leq 6 \exp \left(  - k_2 n \min (\eta, \eta^2)  \right). 
$$

The result follows by setting $ \eta = k_3  \sqrt{\frac{\log(n) + 2\log(p) + \log(q)}{n}} $ for a large enough $ k_3 > 0 $, and taking the union over the $ \pi = n p^2 q $ possible choices of $ i, j, h, l $.  
\end{proof}

\begin{lemma}\label{lemma_KKT}
Let $ \widehat{\mathbf{\Theta}} $ be defined as in \eqref{eq_estimation}. Then, under the assumptions of Theorem~\ref{thm_pred_error}:
\begin{equation}
\sum_{l = \widehat{t}_j}^{n} Y_{l-1} \left( y_l^\prime - Y_{l-1}^\prime \sum_{i=1}^{l} \widehat{\theta}_i^\prime \right) = \frac{n \lambda_n}{2} \mbox{sign} (\widehat{\theta}_{\widehat{t}_j}^\prime), \hspace{1cm} \mbox{for} \,\,\,  j = 1, 2, ..., \widehat{m},
\end{equation}
where $ Y_l^\prime = \left( y_l^\prime \ldots y_{l-q+1}^\prime  \right)_{1 \times pq} $, and 
\begin{equation}
\left|\left|  \sum_{l = j}^{n} Y_{l-1} \left( y_l^\prime - Y_{l-1}^\prime \sum_{i=1}^{l} \widehat{\theta}_i^\prime \right) \right|\right|_\infty \leq \frac{n \lambda_n}{2}, \hspace{1cm} \mbox{for} \,\,\,  j = q-1, 2, ..., n.
\end{equation}
Moreover, $ \sum_{i=1}^{t} \widehat{\theta}_i = \widehat{\Phi}^{(.,j)} $ for $ \widehat{t}_{j-1} \leq t \leq \widehat{t}_j - 1$, $ j = 1, 2, ..., |  \mathcal{A}_n| $. 
\end{lemma}
\begin{proof}
The result follows directly from the KKT condition of the optimization problem \eqref{eq_estimation}.
\end{proof}

\begin{lemma}\label{lemma_bound}
Under assumption A1, there exist constants $ c_i > 0 $ such that with probability at least $ 1 - c_1 \exp(-c_2 (\log(q) + 2 \log(p))  ) $, 

\begin{equation}\label{eq:lemma_bound_1}
\sup_{1 \leq j \leq m_0, s \geq t_j, |  t_j - s  |  > n \gamma_n  }   \left|\left|  {(t_j - s)}^{-1} \left( \sum_{l=s}^{t_j - 1} Y_{l-1} Y_{l-1}^\prime - \Gamma_j^q (0) \right)  \right|\right|_\infty \leq c_3 \sqrt{\frac{\log(q) + 2 \log(p)}{n \gamma_n}},
\end{equation}
where $ \Gamma_j^q (0) = \mathbb{E} (Y_{l-1} Y_{l-1}^\prime) $. Moreover,  
\begin{equation}\label{eq:lemma_bound_2}
\sup_{1 \leq j \leq m_0, s \geq t_j, |  t_j - s  |  > n \gamma_n  }  \left|\left|  {(t_j - s)}^{-1} \sum_{l =  s }^{t_j - 1} Y_{l-1} \varepsilon_l^\prime  \right|\right|_\infty \leq c_3 \sqrt{\frac{\log(q) + 2 \log(p)}{n \gamma_n}}.
\end{equation}

\end{lemma}

\begin{proof} 
The proof of this lemma is similar to that of Proposition~2.4 in \cite{Basu_2015}. Here we briefly outline the main steps of the proof, while omitting the details. For \eqref{eq:lemma_bound_1}, note that using an argument similar to  Proposition~2.4(a) in \cite{Basu_2015}, there exist $ k_1, k_2 > 0 $ such that for each fixed $ k, l = 1, \cdots, pq $,
\begin{equation}
\mathbb{P} \left( \left|  e_k^\prime  \frac{ \sum_{l=s}^{t_j - 1} Y_{l-1} Y_{l-1}^\prime - \Gamma_j^q (0) }{t_j - s} e_l  \right| > k_1 \eta \right) \leq 6 \exp (-k_2 n \gamma_n \min(\eta, \eta^2) ).
\end{equation}
Setting $ \eta = k_3 \sqrt{\frac{\log(q p^2)}{n \gamma_n}} $, and taking union over all possible values of $ k, l $, we obtain \eqref{eq:lemma_bound_1}. 

The proof for \eqref{eq:lemma_bound_2}, is similar to Lemma~\ref{lemma_first}. Again, there exist $ k_1, k_2 > 0 $ such that for each fixed $ k = 1, ..., pq $, $ l = 1, ..., p $,
\begin{equation}
\mathbb{P} \left( \left|  e_k^\prime  \frac{ \sum_{l =  s }^{t_j - 1} Y_{l-1} \varepsilon_l^\prime  }{t_j - s} e_l  \right| > k_1 \eta \right) \leq 6 \exp (-k_2 n \gamma_n \min(\eta, \eta^2) ).
\end{equation}
Setting $ \eta = k_3 \sqrt{\frac{\log(q p^2)}{n \gamma_n}} $, and taking union over all possible values of $ k, l $, we get:
\begin{equation}
 \left|\left|  {(t_j - s)}^{-1} \left( \sum_{l=s}^{t_j - 1} Y_{l-1} Y_{l-1}^\prime - \Gamma_j^q (0) \right)  \right|\right|_\infty \leq c_3 \sqrt{\frac{\log(q) + 2 \log(p)}{n \gamma_n}},
\end{equation}
and 
\begin{equation}
 \left|\left|  {(t_j - s)}^{-1} \sum_{l =  s }^{t_j - 1} Y_{l-1} \varepsilon_l^\prime  \right|\right|_\infty \leq c_3 \sqrt{\frac{\log(q) + 2 \log(p)}{n \gamma_n}},
\end{equation}
with high probability converging to 1 for any $ j = 1, 2, \cdots, m_0 $, as long as $ |  t_j - s  |  > n \gamma_n $ and $ s \geq t_{j-1} $. 
Note that the constants $ c_1, c_2$ and $c_3 $ can be chosen large enough  such that the upper bounds above would be independent of the break point $ t_i $. Therefore, we have the desired upper bounds verified with probability at least $ 1 - c_1 \exp(-c_2 (\log(q) + 2 \log(p))  ) $. 
\end{proof}

\begin{lemma}\label{lemma_selection}
Under the assumptions of Theorem~\ref{thm_selection}, for $ m < m_0 $, there exist constants $ c_1, c_2 > 0 $ such that:

\begin{equation}\label{eq_lower_bound}
\mathbb{P} \left( \min_{(s_1, ..., s_m) \subset \lbrace 1, ..., T \rbrace} L_n(s_1, s_2, ..., s_m; \eta_n)  >  \sum_{t=q}^{T} || \varepsilon_t ||_2^2  + c_1 \Delta_n - c_2 m n \gamma_n {d_n^\star}^2 \right) \rightarrow 1,
\end{equation}
where $ \Delta_n =  \min_{1 \leq j \leq m_0+1} | t_j - t_{j-1} |  $. 
\end{lemma}
\begin{proof} 
Since $ m < m_0 $, there exists a point $ t_j $ such that $ | s_i - t_j | > \Delta_n/4 $. In order to find a lower bound on the sum of the least squares, we consider three different cases: (a) $ | s_i - s_{i-1} | \leq n \gamma_n $; (b) there exist two true break points $ t_j, t_{j+1} $ such that $ | s_{i-1} - t_j | \leq n \gamma_n $ and $ | s_i - t_{j+1} | \leq n \gamma_n $; and (c) otherwise. The idea is to find a lower bound for the sum of squared errors plus the penalty term for each case. Here, we consider only one candidate for each case. The general case can be argued similarly, but omitted here to avoid complex notations. 

Denote the estimated parameter in each of the estimated segments below by $ \widehat{\theta} $.
For case (a), consider the case where the interval $ (s_{i-1}, s_i) $ is inside a true segment. In other words, suppose there exists $ j $ such that $ t_j < s_{i-1} < s_i < t_{j+1} $. Now, 
\begin{eqnarray}
\sum_{t=s_{i-1}}^{s_i-1} || y_t - \widehat{\theta} Y_{t-1}   ||_2^2 &=& \sum_{t=s_{i-1}}^{s_i-1} || \varepsilon_t ||_2^2 + \sum_{t=s_{i-1}}^{s_i-1} || (\Phi^{(.,j+1)} -  \widehat{\theta} )  Y_{t-1}  ||_2^2 \nonumber \\
& + & 2 \sum_{t=s_{i-1}}^{s_i-1} Y_{t-1}^\prime (\Phi^{(.,j+1)} -  \widehat{\theta} )^\prime \varepsilon_t \nonumber \\
& \geq & \sum_{t=s_{i-1}}^{s_i-1} || \varepsilon_t ||_2^2 - \left|  2 \sum_{t=s_{i-1}}^{s_i-1} Y_{t-1}^\prime (\Phi^{(.,j+1)} -  \widehat{\theta} )^\prime \varepsilon_t \right| \nonumber \\ 
& \geq & \sum_{t=s_{i-1}}^{s_i-1} || \varepsilon_t ||_2^2 - c \sqrt{n \gamma_n \log p} || \Phi^{(.,j+1)} -  \widehat{\theta} ||_1.
\end{eqnarray}
Therefore, given the tuning parameter selected based on Assumption~A4, we have:
\begin{equation}
\sum_{t=s_{i-1}}^{s_i-1} || y_t - \widehat{\theta} Y_{t-1}   ||_2^2 + \eta_{(s_{i-1}, s_i)} || \widehat{\theta} ||_1 \geq \sum_{t=s_{i-1}}^{s_i-1} || \varepsilon_t ||_2^2 - c \sqrt{n \gamma_n \log p} || \Phi^{(.,j+1)} ||_1.
\end{equation}

For case (b), consider the case where $ s_{i-1} < t_j $ and $ s_i < t_{j+1} $. Now, similar arguments as in Proposition~4.1 of \citet{Basu_2015} show that by the tuning parameter selected based on A4(b), we have:
\begin{equation}\label{eq_converge}
|| \Phi^{(.,j+1)} -  \widehat{\theta} ||_1 \leq 4 \sqrt{d_n^\star} || \Phi^{(.,j+1)} -  \widehat{\theta} ||_2, \,\,\, \mbox{and} \,\,\, || \Phi^{(.,j+1)} -  \widehat{\theta} ||_2 \leq c_3 \sqrt{d_n^\star} \eta_{(s_{i-1}, s_i)}.
\end{equation}
To see this, observe that $ \widehat{\theta} $ in \eqref{eq_estimation_second} minimizes the least squares plus the $ \ell_1 $ norm loss function. Therefore, the value of this objective function for $ \widehat{\theta} $ will be smaller than any other choice of parameters, including $ \Phi^{(.,j+1)} $. Hence, 
\begin{eqnarray}
\frac{1}{s_i - s_{i-1}} \sum_{t=s_{i-1}}^{s_i-1} || y_t - \widehat{\theta} Y_{t-1}   ||_2^2 + \eta_{(s_{i-1}, s_i)} || \widehat{\theta} ||_1 &\leq & \frac{1}{s_i - s_{i-1}} \sum_{t=s_{i-1}}^{s_i-1} || y_t - \Phi^{(.,j+1)} Y_{t-1}   ||_2^2 \nonumber \\
& + &  \eta_{(s_{i-1}, s_i)} || \Phi^{(.,j+1)} ||_1.
\end{eqnarray}
Some rearrangements lead to:
\begin{eqnarray}\label{eq_Basu}
0 \leq c^\prime ||  \Phi^{(.,j+1)} -  \widehat{\theta} ||_2^2 
& \leq & \frac{1}{s_i - s_{i-1}} \sum_{t=s_{i-1}}^{s_i-1} Y_{t-1}^\prime {\left( \Phi^{(.,j+1)} -  \widehat{\theta} \right)}^\prime \left( \Phi^{(.,j+1)} -  \widehat{\theta} \right) Y_{t-1} \nonumber \\
& \leq & \frac{2}{s_i - s_{i-1}} \sum_{t=s_{i-1}}^{s_i-1} Y_{t-1}^\prime {\left( \Phi^{(.,j+1)} -  \widehat{\theta} \right)}^\prime \left( y_t - \Phi^{(.,j+1)} Y_{t-1} \right) \nonumber \\
& \hspace{1cm} + &  \eta_{(s_{i-1}, s_i)} \left( || \Phi^{(.,j+1)} ||_1 - || \widehat{\theta} ||_1 \right) \nonumber \\
& \leq & \left( c \sqrt{\frac{\log p}{s_i - s_{i-1}}} + M_\Phi d_n^\star \frac{n \gamma_n}{s_i - s_{i-1}}  \right) || \Phi^{(.,j+1)} -  \widehat{\theta} ||_1 \nonumber \\
& \hspace{1cm} + & \eta_{(s_{i-1}, s_i)} \left( || \Phi^{(.,j+1)} ||_1 - || \widehat{\theta} ||_1 \right) \nonumber \\
& \leq & \frac{\eta_{(s_{i-1}, s_i)}}{2} || \Phi^{(.,j+1)} -  \widehat{\theta} ||_1 + \eta_{(s_{i-1}, s_i)} \left( || \Phi^{(.,j+1)} ||_1 - || \widehat{\theta} ||_1 \right) \nonumber \\
& \leq & \frac{3 \eta_{(s_{i-1}, s_i)}}{2} || \Phi^{(.,j+1)} -  \widehat{\theta} ||_{1,\mathcal{I}} - \frac{ \eta_{(s_{i-1}, s_i)}}{2} || \Phi^{(.,j+1)} -  \widehat{\theta} ||_{1,\mathcal{I}^c} \nonumber \\
& \leq & 2 \eta_{(s_{i-1}, s_i)} || \Phi^{(.,j+1)} -  \widehat{\theta} ||_{1}.
\end{eqnarray}
This ensures that $ || \Phi^{(.,j+1)} -  \widehat{\theta} ||_{1,\mathcal{I}^c} \leq 3 || \Phi^{(.,j+1)} -  \widehat{\theta} ||_{1,\mathcal{I}} $, and hence $ || \Phi^{(.,j+1)} -  \widehat{\theta} ||_{1} \leq 4 || \Phi^{(.,j+1)} -  \widehat{\theta} ||_{1,\mathcal{I}} \leq 4 \sqrt{d_n^\star} || \Phi^{(.,j+1)} -  \widehat{\theta} ||_{2} $. This comparison between $ L_1 $ and $ L_2 $ norms of the error term together with the bound in Equation~\ref{eq_Basu} will get the desired consistency rates in \eqref{eq_converge}. 

Similar to case (a), using lemma (\ref{lemma_bound}), we have:
\begin{eqnarray}\label{eq:31}
\sum_{t=t_j}^{s_i-1} || y_t - \widehat{\theta} Y_{t-1}   ||_2^2 & \geq & \sum_{t=t_j}^{s_i-1} || \varepsilon_t ||_2^2 + c | s_i - t_j | \, || \Phi^{(.,j+1)} -  \widehat{\theta} ||_2^2 - c^\prime \sqrt{| s_i - t_j | \log p} \, || \Phi^{(.,j+1)} -  \widehat{\theta} ||_1 \nonumber \\
& \geq &  \sum_{t=t_j}^{s_i-1} || \varepsilon_t ||_2^2 + c | s_i - t_j | \, || \Phi^{(.,j+1)} -  \widehat{\theta} ||_2 \left( || \Phi^{(.,j+1)} -  \widehat{\theta} ||_2 - \frac{c^\prime}{c} \sqrt{\frac{ d_n^\star \log p }{|s_i - t_j|}} \right) \nonumber \\
& \geq & \sum_{t=t_j}^{s_i-1} || \varepsilon_t ||_2^2 - c^\prime d_n^\star \log p.
\end{eqnarray}
Also, for the interval $ ( s_{i-1}, t_j ) $, we have:
\begin{eqnarray}\label{eq:32}
\sum_{t=s_{i-1}}^{t_j-1} || y_t - \widehat{\theta} Y_{t-1}   ||_2^2 & \geq & \sum_{t=s_{i-1}}^{t_j-1} || \varepsilon_t ||_2^2 - c^\prime \sqrt{n \gamma_n \log p} \, || \Phi^{(.,j)} -  \widehat{\theta} ||_1 \nonumber \\
& \geq & \sum_{t=s_{i-1}}^{t_j-1} || \varepsilon_t ||_2^2 - c^\prime \sqrt{n \gamma_n \log p} \, \left( || \Phi^{(.,j+1)} -  \widehat{\theta} ||_1 + || \Phi^{(.,j+1)} -  \Phi^{(.,j)} ||_1 \right) \nonumber \\
& \geq & \sum_{t=s_{i-1}}^{t_j-1} || \varepsilon_t ||_2^2 - c^\prime \sqrt{n \gamma_n \log p} \, \left( d_n^\star \eta_{(s_{i-1}, s_i)} + || \Phi^{(.,j+1)} -  \Phi^{(.,j)} ||_1 \right) \nonumber \\
& \geq & \sum_{t=s_{i-1}}^{t_j-1} || \varepsilon_t ||_2^2 - c^\prime  d_n^\star \sqrt{n \gamma_n \log p}.
\end{eqnarray}
Combining equations \eqref{eq:31} and \eqref{eq:32} gives:
\begin{equation}
\sum_{t=s_{i-1}}^{s_i-1} || y_t - \widehat{\theta} Y_{t-1}   ||_2^2  \geq  \sum_{t=s_{i-1}}^{s_i-1} || \varepsilon_t ||_2^2 - c^\prime  d_n^\star \sqrt{n \gamma_n \log p}.
\end{equation}

For case (c), consider the case where $ s_{i-1} < t_j < s_i $ with $ |s_{i-1} - t_j| > \Delta_n/4 $ and $ |s_{i} - t_j| > \Delta_n/4 $. Similar arguments as in Proposition~4.1 of \citet{Basu_2015} shows that :
\begin{equation}
|| \Phi^{(.,j+1)} -  \widehat{\theta} ||_1 \leq 4 \sqrt{d_n^\star} || \Phi^{(.,j+1)} -  \widehat{\theta} ||_2, \,\,\, \mbox{and} \,\,\, || \Phi^{(.,j)} -  \widehat{\theta} ||_1 \leq 4 \sqrt{d_n^\star} || \Phi^{(.,j)} -  \widehat{\theta} ||_2.
\end{equation}

Note that in this case, the restricted eigenvalue condition does not hold. Therefore, the convergence of the $ \widehat{\theta} $ cannot be verified. The reason is that in this case, the two parts of the true segments which intersect with the estimated segment have large lengths. If the length of one of them was negligible as compared to the other segment, one could still verify the restricted eigenvalue, but that's not the case here. However, the deterministic part of the deviation bound argument holds with the suitable choice of the tuning parameter. Now, similar to case (b), on both intervals $ (s_{i-1},t_j) $ and $ (t_j, s_i) $:
\begin{eqnarray}\label{eq:consequtive}
\sum_{t=s_{i-1}}^{t_j-1} || y_t - \widehat{\theta} Y_{t-1}   ||_2^2 
& \geq & \sum_{t=s_{i-1}}^{t_j-1} || \varepsilon_t ||_2^2 + c | t_j - s_{i-1} | \, || \Phi^{(.,j)} -  \widehat{\theta} ||_2^2 \nonumber \\
& \hspace{1cm} - & c^\prime \sqrt{| t_j - s_{i-1} | \log p} \, || \Phi^{(.,j)} -  \widehat{\theta} ||_1 \nonumber \\
& \geq &  \sum_{t=s_{i-1}}^{t_j-1} || \varepsilon_t ||_2^2 \nonumber \\
 & \hspace{-2.5cm} + & \hspace{-2cm} c | t_j - s_{i-1} | \, || \Phi^{(.,j)} -  \widehat{\theta} ||_2 \left( || \Phi^{(.,j)} -  \widehat{\theta} ||_2 - \frac{c^\prime}{c} \sqrt{\frac{ d_n^\star \log p }{|t_j - s_{i-1}|}} \right),
\end{eqnarray}
and 
\begin{eqnarray}
\sum_{t=t_j}^{s_i-1} || y_t - \widehat{\theta} Y_{t-1}   ||_2^2 & \geq & \sum_{t=t_j}^{s_i-1} || \varepsilon_t ||_2^2 + c | s_i - t_j | \, || \Phi^{(.,j+1)} -  \widehat{\theta} ||_2^2 - c^\prime \sqrt{| s_i - t_j | \log p} \, || \Phi^{(.,j+1)} -  \widehat{\theta} ||_1 \nonumber \\
& \hspace{-2cm} \geq & \hspace{-1cm} \sum_{t=t_j}^{s_i-1} || \varepsilon_t ||_2^2 + c | s_i - t_j | \, || \Phi^{(.,j+1)} -  \widehat{\theta} ||_2 \left( || \Phi^{(.,j+1)} -  \widehat{\theta} ||_2 - \frac{c^\prime}{c} \sqrt{\frac{ d_n^\star \log p }{|s_i - t_j|}} \right).
\end{eqnarray}

Since $ || \Phi^{(.,j+1)} -  \Phi^{(.,j)} ||_2 \geq v > 0 $, either $ || \Phi^{(.,j+1)} -  \widehat{\theta} ||_2 \geq v/4 $ or $ || \Phi^{(.,j)} -  \widehat{\theta} ||_2 \geq v/4 $. Assume that $ || \Phi^{(.,j)} -  \widehat{\theta} ||_2 \geq v/4 $. Then, based on Equation~\ref{eq:consequtive}, for some $ c_1 > 0 $,
\begin{equation}\label{eq:37}
\sum_{t=s_{i-1}}^{t_j-1} || y_t - \widehat{\theta} Y_{t-1}   ||_2^2  \geq \sum_{t=s_{i-1}}^{t_j-1} || \varepsilon_t ||_2^2 + c_1 \Delta_n.
\end{equation}
For the second interval we have:
\begin{equation}\label{eq:38}
\sum_{t=t_j}^{s_i-1} || y_t - \widehat{\theta} Y_{t-1}   ||_2^2 \geq  \sum_{t=t_j}^{s_i-1} || \varepsilon_t ||_2^2 - c^\prime d_n^\star \log p.
\end{equation}
Combining \eqref{eq:37} and \eqref{eq:38}, leads to:
\begin{equation}
\sum_{t=s_{i-1}}^{s_i-1} || y_t - \widehat{\theta} Y_{t-1}   ||_2^2  \geq \sum_{t=s_{i-1}}^{s_i-1} || \varepsilon_t ||_2^2 + c_1 \Delta_n - c^\prime d_n^\star \log p.
\end{equation}

Note that another situation may arise in this case, where $ |s_{i-1} - t_j| > n \gamma_n $ and $ |s_{i} - t_j| > n \gamma_n $. Using similar augments as above in this situation, we get the following lower bound:
\begin{equation}
\sum_{t=s_{i-1}}^{s_i-1} || y_t - \widehat{\theta} Y_{t-1}   ||_2^2  \geq \sum_{t=s_{i-1}}^{s_i-1} || \varepsilon_t ||_2^2 - c^\prime {d_n^\star}^2 n \gamma_n.
\end{equation}

Putting all of the cases together will yield the result.

\end{proof}

\subsection*{Appendix~B: Proof of Main Results}

\begin{proof}[Proof of Theorem~\ref{thm_pred_error}]
By definition of $ \widehat{\Theta} $, we get
\begin{eqnarray}
\frac{1}{n} || \textbf{Y} - \textbf{Z} \widehat{\mathbf{\Theta}}  ||_2^2 + \lambda_{1,n} \sum_{i=1}^{n} || \widehat{\theta}_i ||_1 + \lambda_{2,n} \sum_{k=1}^{n} \left \| \sum_{j=1}^{k} \widehat{\theta}_j  \right \|_1 
& &  \nonumber \\
& \hspace{-9cm} \leq & \hspace{-4.5cm} \frac{1}{n} || \textbf{Y} - \textbf{Z} \mathbf{\Theta} ||_2^2 + \lambda_{1,n} \sum_{i=1}^{n} || \theta_i ||_1 + \lambda_{2,n} \sum_{k=1}^{n} \left \| \sum_{j=1}^{k} \theta_j \right \|_1.
\end{eqnarray}

Denoting $ \mathcal{A} = \lbrace t_1, t_1, \cdots, t_{m_0}  \rbrace $, we have:
\begin{eqnarray}
\frac{1}{n} \left|\left| \textbf{Z} \left( \widehat{\mathbf{\Theta}} - \mathbf{\Theta}  \right)  \right|\right|_2^2 &\leq & \frac{2}{n} \left(  \widehat{\mathbf{\Theta}} - \mathbf{\Theta} \right)^\prime \textbf{Z}^\prime \textbf{E} + \lambda_{1,n} \sum_{i=1}^{n} || {\theta}_i ||_1 - \lambda_{1,n} \sum_{i=1}^{n} || \widehat{\theta}_i ||_1 \nonumber \\ 
& + & \lambda_{2,n} \sum_{k=1}^{n} \left \| \sum_{j=1}^{k} \theta_j \right \|_1 - \lambda_{2,n} \sum_{k=1}^{n} \left \| \sum_{j=1}^{k} \widehat{\theta}_j  \right \|_1 \nonumber \\
& \leq & 2 \left|\left|  \frac{\textbf{Z}^\prime \textbf{E} }{n}  \right|\right|_\infty \sum_{i=1}^{n}  || \theta_i - \widehat{\theta}_i ||_1 + \lambda_{1,n} \sum_{i \in \mathcal{A}} \left( || \theta_i ||_1 -  || \widehat{\theta}_i ||_1  \right) - \lambda_{1,n} \sum_{i \in \mathcal{A}^c} || \widehat{\theta}_i ||_1 \nonumber \\ 
& + & \lambda_{2,n} \sum_{j=1}^{m_0+1} (t_j - t_{j-1}) \| \Phi^{(.,j)} \|_1 \nonumber \\ 
& \leq & \lambda_{1,n} \sum_{i \in \mathcal{A}}  || \theta_i - \widehat{\theta}_i ||_1 + \lambda_{1,n} \sum_{i \in \mathcal{A}} \left( || \theta_i ||_1 -  || \widehat{\theta}_i ||_1  \right) + \lambda_{2,n} n \, d_n^\star \nonumber \\
& \leq & 2 \lambda_{1,n} \sum_{i \in \mathcal{A}} || \theta_i ||_1 + + \lambda_{2,n} n \, d_n^\star \nonumber \\
& \leq & 2 \lambda_{1,n} m_n  \max_{1 \leq j \leq m_0+1} \left|\left| \Phi^{(.,j)} - \Phi^{(.,j-1)} \right|\right|_1 + o(1) \nonumber \\
& \hspace{-5cm} \leq & \hspace{-2.5cm} 4 C m_n  \max_{1 \leq j \leq m_0+1} \left\lbrace \sum_{k=1}^{p} \left( d_{kj} + d_{k(j-1)}  \right) \right\rbrace   M_\Phi \sqrt{\frac{\log(n) + 2\log(p) + \log(q)}{n}} + o(1),
\end{eqnarray}
with high probability approaching to 1 due to Lemma~\ref{lemma_first}. 
\end{proof}

\begin{proof}[Proof of Theorem~\ref{thm_Hausdorff}]

{
The proof is different from Theorem~2.2 in \cite{Chan_2014} and Proposition~5 in \cite{Harchaoui_2010} due to the additional penalty added in equation \eqref{eq_estimation}. For a matrix $ A \in \mathbb{R}^{pq \times p} $, let $ || A ||_{1, \mathcal{I}} = \sum_{(j,k) \in \mathcal{I}} |a_{jk}| $. 

First, we focus on the second part. Suppose for some $ j = 1, \cdots, m_0 $, $ | \widehat{t}_j - t_j | > n \gamma_n $. Then, there exists a true break point $ t_{j_0} $ which is isolated from all the estimated points, i.e., $ \min_{1 \leq j \leq m_0} | \widehat{t}_j - t_{j_0} | > n \gamma_n  $. In other words, there exists an estimated break point $ \widehat{t}_j $ such that, $ t_{j_0} - t_{j_0-1} \vee \widehat{t}_j \geq n \gamma_n $ and $ t_{j_0+1} \wedge \widehat{t}_{j+1} - t_{j_0} \geq n \gamma_n $. The idea of the proof is to show the estimated AR parameter estimated in the interval $ [t_{j_0-1} \vee \widehat{t}_j, t_{j_0+1} \wedge \widehat{t}_{j+1}] $ converges in $ L_2 $ to both $ {\Phi}^{(.,j_0)} $ and $ {\Phi}^{(.,j_0+1)} $ which contradicts with assumption A3. This is due to the fact that the length of the interval is large enough to verify restricted eigenvalue and deviation bound inequalities needed to show parameter estimation consistency. 

Based on the definition of $ \widehat{\Theta} $ in \eqref{eq_estimation}, the value of the function defined in \eqref{eq_estimation} is minimized exactly at $ \widehat{\Theta} $. This means that any other choice of parameters yields to higher value in \eqref{eq_estimation}. First, we focus on the interval $ [t_{j_0-1} \vee \widehat{t}_j, t_{j_0}] $. Define a new parameter sequence $ \psi_k $'s, $ k = 1, ..., n $ with $ \psi_k = \widehat{\theta}_k $ except for two time points $ k = \widehat{t}_j $ and $ k = t_{j_0} $. For these two points set $ \psi_{\widehat{t}_j} = \Phi^{(.,j_0)} - \widehat{\Phi}_j $ and $ \psi_{t_{j_0}} = \widehat{\Phi}_{j+1} - \Phi^{(.,j_0)} $ where $ \widehat{\Phi}_j = \sum_{k=1}^{t_{j_0-1} \vee \widehat{t}_j - 1} \widehat{\theta}_k $ and $ \widehat{\Phi}_{j+1} = \sum_{k=1}^{t_{j_0} \vee \widehat{t}_j} \widehat{\theta}_k $, i.e. $ \widehat{\theta}_{t_{j_0} \vee \widehat{t}_j} = \widehat{\Phi}_{j+1} - \widehat{\Phi}_{j} $. Denoting $ \Psi =  \mbox{vector} (\psi_1, ..., \psi_n) \in \mathbb{R}^{\pi \times 1} $, we have

\begin{eqnarray}\label{eq:minimizer}
\frac{1}{n} \| \textbf{Y} - \textbf{Z} \mathbf{\widehat{\Theta}} \|_2^2 + \lambda_{1,n}  \| \mathbf{\widehat{\Theta}} \|_1 + \lambda_{2,n} \sum_{k=1}^{n} \left \| \sum_{j=1}^{k} \widehat{\theta}_j \right \|_1 &\leq & \frac{1}{n} \| \textbf{Y} - \textbf{Z} \mathbf{\Psi} \|_2^2 + \lambda_{1,n}  \| \mathbf{\Psi} \|_1 \\ \nonumber
&+& \lambda_{2,n} \sum_{k=1}^{n} \left \| \sum_{j=1}^{k} \psi_j \right \|_1.
\end{eqnarray}

Some rearrangement of equation \eqref{eq:minimizer} leads to

\begin{eqnarray}\label{eq:newproof}
0 &\leq& c \, \| \Phi^{(.,j_0)} - \widehat{\Phi}_{j+1} \|_2^2 \nonumber \\ & \leq & \frac{1}{t_{j_0} - t_{j_0-1} \vee \widehat{t}_j} \sum_{l=t_{j_0-1} \vee \widehat{t}_j}^{t_{j_0}-1} \left( \Phi^{(.,j_0)} - \widehat{\Phi}_{j+1} \right)^\prime Y_{l-1} Y_{l-1}^\prime \left( \Phi^{(.,j_0)} - \widehat{\Phi}_{j+1} \right) \nonumber \\  
& \leq & \frac{1}{t_{j_0} - t_{j_0-1} \vee \widehat{t}_j} \sum_{l=t_{j_0-1} \vee \widehat{t}_j}^{t_{j_0}-1} Y_{l-1}^\prime \left( \Phi^{(.,j_0)} - \widehat{\Phi}_{j+1} \right) \varepsilon_l \nonumber \\  
&+&  \frac{n \lambda_{1,n}}{t_{j_0} - t_{j_0-1} \vee \widehat{t}_j} \left( \| \Phi^{(.,j_0)} - \widehat{\Phi}_{j+1} \|_1 + \| \Phi^{(.,j_0)} - \widehat{\Phi}_{j} \|_1 - \| \widehat{\Phi}_{j+1} - \widehat{\Phi}_{j} \|_1\right) \nonumber \\
&+& n \lambda_{2,n} \left( \| \Phi^{(.,j_0)} \|_1 - \| \widehat{\Phi}_{j+1} \|_1 \right) \nonumber \\
& \leq &  \left( \frac{2 n \lambda_{1,n}}{t_{j_0} - t_{j_0-1} \vee \widehat{t}_j} + C \sqrt{\frac{\log p}{n \gamma_n}}  \right) \| \Phi^{(.,j_0)} - \widehat{\Phi}_{j+1} \|_1 + n \lambda_{2,n} \left( \| \Phi^{(.,j_0)} \|_1 - \| \widehat{\Phi}_{j+1} \|_1 \right) \nonumber \\
& \leq &  \frac{1}{2}n \lambda_{2,n} \| \Phi^{(.,j_0)} - \widehat{\Phi}_{j+1} \|_1 + n \lambda_{2,n} \left( \| \Phi^{(.,j_0)} \|_1 - \| \widehat{\Phi}_{j+1} \|_1 \right) \nonumber \\
& \leq &  \frac{3}{2}n \lambda_{2,n} \| \Phi^{(.,j_0)} - \widehat{\Phi}_{j+1} \|_{1,\mathcal{I}} - \frac{1}{2}n \lambda_{2,n} \| \Phi^{(.,j_0)} - \widehat{\Phi}_{j+1} \|_{1,\mathcal{I}^c}.
\end{eqnarray}

In equation \eqref{eq:newproof}, the second inequality holds with high probability converging to 1 due to first part of Lemma~\ref{lemma_bound} and the fact that $ t_{j_0} - t_{j_0-1} \vee \widehat{t}_j \geq n \gamma_n $. The fourth inequality holds with high probability converging to 1 due to second part of Lemma~\ref{lemma_bound} and triangular inequality. The fifth inequality is based on the assumption A3 and the selection for $ \lambda_{2,n} $ in the statement of the theorem. The last inequality holds by sparsity assumption. This implies that 

\begin{equation}
\| \Phi^{(.,j_0)} - \widehat{\Phi}_{j+1} \|_2 = o_p \left( d_{n}^\star \sqrt{\frac{\log p}{n \gamma_n}} \right),
\end{equation}
which means that $ \| \Phi^{(.,j_0)} - \widehat{\Phi}_{j+1} \|_2 $ converges to zero in probability based on assumption A3. Similarly, the same procedure can be applied to the interval $ [t_{j_0}, t_{j_0+1} \wedge \widehat{t}_{j+1}] $ which leads to $ \| \Phi^{(.,j_0+1)} - \widehat{\Phi}_{j+1} \|_2 $ converges to zero in probability as well. This yields to a contradiction to the assumption A3, and therefore, the proof is complete.

The proof of the first part is similar to the second part. Hence, a brief sketch is provided. Assume $ |\widehat{\mathcal{A}}_n | < m_0 $. This means there exist an isolated true break point, say $ t_{j_0} $. More specifically, there exists an estimated break point $ \widehat{t}_j $ such that, $ t_{j_0} - t_{j_0-1} \vee \widehat{t}_j \geq n \gamma_n / 3 $ and $ t_{j_0+1} \wedge \widehat{t}_{j+1} - t_{j_0} \geq n \gamma_n / 3 $. Now, similar arguments as explained in details in the second part can be applied to both intervals $ [t_{j_0-1} \vee \widehat{t}_j, t_{j_0}] $ and $ [t_{j_0}, t_{j_0+1} \wedge \widehat{t}_{j+1}] $ which leads to $ \| \Phi^{(.,j_0 + 1)} - \Phi^{(.,j_0)} \|_2 $ converges to zero and therefore contradicts with assumption A3. This completes the proof.

}

\end{proof}

\begin{proof}[Proof of Theorem~\ref{thm_selection}]
For the first part we show that (a) $ \mathbb{P} ( \widetilde{m} < m_0 ) \rightarrow 0 $, and (b) $ \mathbb{P} ( \widetilde{m} > m_0 ) \rightarrow 0 $. For (a), we know, from Theorem~\ref{thm_Hausdorff}, that there exists points $ \widehat{t}_j \in \mathcal{A}_n $ such that $ \max_{1 \leq j \leq m_0} | \widehat{t}_j - t_j  | \leq n \gamma_n $. By similar arguments as in Lemma~\ref{lemma_selection}, we get that there exists a constant $ K > 0 $ such that:
\begin{eqnarray}\label{eq:thm4:m0}
L(\widehat{t}_1, ..., \widehat{t}_{m_0} ;\eta_n) 
& \leq & \sum_{t=q}^{T} || \varepsilon_t ||_2^2 + K m_0 n \gamma_n {d_n^\star}^2.
\end{eqnarray}
To see this, we only show the calculations for one of the estimated segments. Suppose $ s_{i-1} < t_j < s_i $ with $ |t_j - s_{i-1}| \leq n \gamma_n $. Denote the estimated coefficient in the segment $ (s_{i-1}, s_i) $ by $ \widehat{\theta} $. Similar to case (b) in the proof of Lemma~\ref{lemma_selection}, we have:
\begin{eqnarray}\label{eq:57}
\sum_{t=t_j}^{s_i-1} || y_t - \widehat{\theta} Y_{t-1}   ||_2^2 & \leq & \sum_{t=t_j}^{s_i-1} || \varepsilon_t ||_2^2 + c_3 | s_i - t_j | \, || \Phi^{(.,j+1)} -  \widehat{\theta} ||_2^2 \nonumber \\
& + &  c^\prime \sqrt{| s_i - t_j | \log p} \, || \Phi^{(.,j+1)} -  \widehat{\theta} ||_1 \nonumber \\
& \equiv & \sum_{t=t_j}^{s_i-1} || \varepsilon_t ||_2^2 + I + II.
\end{eqnarray}

Now, by the convergence rate of the error (see, e.g., case (b) in the proof of Lemma~\ref{lemma_selection}),
\begin{eqnarray}\label{eq:58}
I & \leq & 4 c_3 | s_i - t_j | d_n^\star {\left( c \sqrt{\frac{\log p}{| s_i - t_j |}} + M_\Phi d_n^\star \frac{n \gamma_n}{| s_i - t_j |} \right)}^2 \nonumber \\
& = & O_p \left( n \gamma_n {d_n^\star}^2 \right),
\end{eqnarray}
and
\begin{eqnarray}\label{eq:59}
II & \leq & c^\prime \sqrt{| s_i - t_j | \log p} \, d_n^\star \left(  c \sqrt{\frac{\log p}{| s_i - t_j |}} + M_\Phi d_n^\star \frac{n \gamma_n}{| s_i - t_j |} \right) \nonumber \\
& = & O_p \left( n \gamma_n {d_n^\star}^2 \right).
\end{eqnarray}
Applying a similar argument to the smaller sub-segment $ (s_{i-1}, t_j) $, we get:
\begin{eqnarray}\label{eq:60}
\sum_{t=s_{i-1}}^{t_j-1} || y_t - \widehat{\theta} Y_{t-1}   ||_2^2 & \leq & \sum_{t=s_{i-1}}^{t_j-1} || \varepsilon_t ||_2^2 + c_3 | t_j - s_{i-1} | \, || \Phi^{(.,j)} -  \widehat{\theta} ||_2^2 \nonumber \\
& + & c^\prime \sqrt{| t_j - s_{i-1} | \log p} \, || \Phi^{(.,j)} -  \widehat{\theta} ||_1 \nonumber \\
& \leq & \sum_{t=s_{i-1}}^{t_j-1} || \varepsilon_t ||_2^2 + 2 c_3 | t_j - s_{i-1} | \, \left( || \Phi^{(.,j+1)} -  \widehat{\theta} ||_2^2 + || \Phi^{(.,j+1)} - \Phi^{(.,j)} ||_2^2 \right)  \nonumber \\
& + & c^\prime \sqrt{| t_j - s_{i-1} | \log p} \, \left( || \Phi^{(.,j+1)} -  \widehat{\theta} ||_1 + || \Phi^{(.,j+1)} -  \Phi^{(.,j)} ||_1 \right) \nonumber \\
& = & \sum_{t=s_{i-1}}^{t_j-1} || \varepsilon_t ||_2^2 + O_p \left( n \gamma_n {d_n^\star}^2 \right).
\end{eqnarray}

Finally,
\begin{eqnarray}\label{eq:61}
\eta_{(s_{i-1}, s_i)} || \widehat{\theta} ||_1 & \leq & \eta_{(s_{i-1}, s_i)} \left( || \Phi^{(.,j+1)} -  \widehat{\theta} ||_1 + || \Phi^{(.,j+1)} ||_1  \right) \nonumber \\
& = & O_p (d_n^\star).
\end{eqnarray}

Combining \eqref{eq:57}--\eqref{eq:61} leads to:
\begin{equation}
\sum_{t=s_{i-1}}^{s_i-1} || y_t - \widehat{\theta} Y_{t-1}   ||_2^2 + \eta_{(s_{i-1}, s_i)} || \widehat{\theta} ||_1 = \sum_{t=s_{i-1}}^{s_i-1} || \varepsilon_t ||_2^2 + O_p \left( n \gamma_n {d_n^\star}^2 \right).
\end{equation}

Adding these equations over all $ m_0 + 1 $ segments leads to equation~\ref{eq:thm4:m0}.

Now, applying Lemma~\ref{lemma_selection}, we get:
\begin{eqnarray}
IC( \widetilde{t}_1, ..., \widetilde{t}_{\widetilde{m}} ) & = & L_n (\widetilde{t}_1, ..., \widetilde{t}_{\widetilde{m}}; \eta_n ) + \widetilde{m} \omega_n \nonumber \\
& > & \sum_{t=q}^{T} || \varepsilon_t ||_2^2 + c_1 \Delta_n - c_2 \widetilde{m} n \gamma_n {d_n^\star}^2 + \widetilde{m} \, \omega_n \nonumber \\
& \geq & L(\widehat{t}_1, ..., \widehat{t}_{m_0} ;\eta_n) + m_0 \omega_n + c_1 \Delta_n - c_2 m_0 n \gamma_n d_n^\star - (m_0 -  \widetilde{m}) \omega_n \nonumber \\
& \geq & L(\widehat{t}_1, ..., \widehat{t}_{m_0} ;\eta_n) + m_0 \omega_n,
\end{eqnarray}
since $ \lim_{n \rightarrow \infty} n \gamma_n {d_n^\star}^2 / \omega_n \leq 1 $, and $ \lim_{n \rightarrow \infty} m_0 \omega_n / \Delta_n = 0 $. This proves part (a). To prove part (b), note that a similar argument as in Lemma~\ref{lemma_selection} shows that 
\begin{equation}
L_n ( \widetilde{t}_1, ..., \widetilde{t}_{\widetilde{m}}  ; \eta_n) \geq \sum_{t=q}^{T} || \varepsilon_t ||_2^2 - c_2 \widetilde{m} n \gamma_n {d_n^\star}^2.
\end{equation}
A comparison between $ IC (\widetilde{t}_1, ..., \widetilde{t}_{\widetilde{m}}) $ and $ IC (\widehat{t}_1, ..., \widehat{t}_{m_0}) $ leads to:
\begin{eqnarray}
\sum_{t=q}^{T} || \varepsilon_t ||_2^2 - c_2 \widetilde{m} n \gamma_n {d_n^\star}^2 + m \omega_n & \leq & IC (\widetilde{t}_1, ..., \widetilde{t}_{\widetilde{m}}) \nonumber \\
& \leq & IC (\widehat{t}_1, ..., \widehat{t}_{m_0}) \nonumber \\
& \leq & \sum_{t=q}^{T} || \varepsilon_t ||_2^2 + K m_0 n \gamma_n {d_n^\star}^2 + m_0 \omega_n,
\end{eqnarray}
which means
\begin{equation}\label{eq:unnumbered1}
(\widetilde{m} - m_0) \omega_n \leq c_2 \widetilde{m} n \gamma_n {d_n^\star}^2 + K m_0 n \gamma_n {d_n^\star}^2. 
\end{equation}
However, \eqref{eq:unnumbered1} contradicts with the fact that $ m_0 n \gamma_n {d_n^\star}^2 / \omega_n \rightarrow 0 $. This completes the first part of the theorem. 

For the second part, let $ B = 2 K / c $, and suppose there exists a point $ t_i $ such that $ \min_{1 \leq j \leq m_0} | \widetilde{t}_j - t_j | \geq B m_0 n \gamma_n {d_n^\star}^2 $. Then, by similar argument as in Lemma~\ref{lemma_selection}, we can show that:
\begin{eqnarray}
\sum_{t=d}^{T} || \varepsilon_t ||_2^2 + c B m_0 n \gamma_n {d_n^\star}^2 & < & L_n ( \widetilde{t}_1, ..., \widetilde{t}_{m_0} ) \nonumber \\
& \leq & L_n ( \widehat{t}_1, ..., \widehat{t}_{m_0} ) \nonumber \\
& \leq & \sum_{t=q}^{T} || \varepsilon_t ||_2^2 + K m_0 n \gamma_n {d_n^\star}^2, 
\end{eqnarray} 
which contradicts with the way $ B $ was selected. This completes the proof of the theorem.
\end{proof}

\begin{proof}[Proof of Theorem~\ref{thm_parameter_consistency}]

The proof of this theorem is similar to that of Proposition~4.1 in \citet{Basu_2015}. The two main components of the proof is (i) verifying the restricted eigenvalue (RE) for $ \hat{\Gamma} = I_p \otimes (\mathcal{X}_{\textbf{r}}^\prime \mathcal{X}_{\textbf{r}}/N) $, and (ii) verifying the deviation bound for $ \left|\left| \hat{\gamma} - \hat{\Gamma} \Phi \right|\right|_\infty $ where $ \hat{\gamma} = (I_p \otimes \mathcal{X}_{\textbf{r}}^\prime) \textbf{Y}_{\textbf{r}}/N $. Once these two are verified, the rest of the proof is applying deterministic arguments used in Proposition~4.1 in \citet{Basu_2015}. Therefore, here we proof (i) and (ii) only. 

Condition (i) means that there exist $ \alpha, \tau > 0 $ such that for any $ \theta \in \mathbb{R}^{\tilde{\pi}} $, we have 
$$ \theta^\prime \hat{\Gamma} \theta \geq \alpha || \theta ||_2^2 - \tau || \theta ||_1^2, $$ with probability at least $ 1 - c_1 \exp(- c_2 N) $ for large enough constants $ c_1, c_2 > 0 $. Based on Lemma~B.1 in \citet{Basu_2015}, it is enough to show the RE for $ S = \mathcal{X}_i^\prime \mathcal{X}_i/N $, where $ \mathcal{X}_i $ is the $i$th block component of $ \mathcal{X}_{\textbf{r}} $. Applying Proposition~2.4 in \citet{Basu_2015}, we have for any $ v \in \mathbb{R}^{pq} $ with $ ||v||_2 \leq 1 $, and any $ \eta > 0 $:
$$ \mathbb{P} \left( \left| v^\prime \left(S - \frac{N_i}{N} \Gamma_i(0) \right) v \right| > c \eta    \right) \leq 2 \exp (- c_3 N \min(\eta^2,\eta) ). $$ Now, to make the above probability hold uniformly on all the vectors $ v $, we apply the discretization Lemma~F2 in \citet{Basu_2015} and also Lemma~12 in the supplementary materials of \citet{loh2012} to get:

$$ \left| v^\prime \left(S - \frac{N_i}{N} \Gamma_i(0) \right) v \right| \leq \alpha ||v||_2^2 + \alpha/k ||v||_1^2, $$

with high probability at least $ 1 - c_1 \exp(- c_2 N)  $, for all $ v \in \mathbb{R}^{pq} $, some $ \alpha > 0 $ and with an integer $ k = \lceil{ c_4 N/\log (pq)}\rceil $ with some $ c_4 > 0 $. This implies that 
$$ v^\prime S v \geq v^\prime \frac{N_i}{N} \Gamma_i(0)v - \alpha ||v||_2^2 - \alpha/k ||v||_1^2 \geq \alpha ||v||_2^2 - \alpha/k ||v||_1^2, $$
since $ N_i \geq \Delta_n - 4 R_n $, $ N = n + q - 1 - 2 m_0 R_n $, and assuming $ \Delta_n \geq \varepsilon n $ implies that $ N_i / N \geq \varepsilon \geq 2 \alpha $. 

The deviation condition (DC) here means that there exist a large enough constant $ C^\prime > 0 $ such that 
$$ \left|\left| \hat{\gamma} - \hat{\Gamma} \Phi \right|\right|_\infty \leq C^\prime \sqrt{\frac{\tilde{\pi}}{N}},  $$
with probability at least $ 1 - c_1 \exp( -c_2 \log \tilde{q} ) $. To verify this condition here, observe that $ \hat{\gamma} - \hat{\Gamma} \Phi = \mbox{vec} \left( \mathcal{X}_{\textbf{r}}^\prime E_{\textbf{r}} \right)/N $. Therefore, denoting the $ h$--th column block of $ \mathcal{X}_{\textbf{r}} $ by $ \mathcal{X}_{\textbf{r},(h)} $, for $ h = 1, ..., (m_0+1)q $, we have:

$$ \left|\left| \hat{\gamma} - \hat{\Gamma} \Phi \right|\right|_\infty =  \max_{1 \leq k,l \leq p; 1 \leq h \leq (m_0+1)d} \left| e_k^\prime  \mathcal{X}_{\textbf{r},(h)}^\prime E_{\textbf{r}} e_l \right|. $$
Now, for a fixed $ k, l, h $, applying Proposition~2.4(b) in \citet{Basu_2015} gives:
$$ \mathbb{P} \left(  \left| e_k^\prime  \mathcal{X}_{\textbf{r},(h)}^\prime E_{\textbf{r}} e_l \right| > k_1 \eta  \right) \leq 6 \exp( - k_2 N \min(\eta^2,\eta)),  $$
for large enough $ k_1, k_2 > 0 $, and any $ \eta > 0 $. Now, setting $ \eta = C^\prime \sqrt{\frac{\tilde{\pi}}{N}} $, and taking the union over all the $ \tilde{\pi} $ cases for $ k, l, h $ yield the desired result. This completes the proof of this theorem.

\end{proof}


\subsection*{Appendix~C: Details of Estimation Algorithms}
In this section, we provide details of the algorithm for solving the optimization problem \eqref{eq_estimation}, as well as the proposed backward elimination algorithm (BEA) for the second-stage screening.  

To describe the algorithm for solving the optimization problem \eqref{eq_estimation}, let $ S(. ; \lambda) $ be the element-wise soft-thresholding operator which maps its input $ x $ to $ x - \lambda $ when $ x > \lambda $, $ x + \lambda $ when $ x < - \lambda $, and $ 0 $ when $ |x| \leq \lambda $. Recall that throughout the paper, for a $ m \times n $ matrix $ A $, $ \| A \|_\infty = \max_{1 \leq i \leq m, 1 \leq j \leq n} |a_{ij}| $. The algorithm is as follows:

\begin{itemize}
\item[(i)] Set the initial values for all parameters to be zero; i.e. $ \theta_i^{(0)} = 0 $, for $ i = 1, \ldots, n $. 
\item[(ii)] For each $ i = 1, \ldots, n $, calculate the $ (h+1)$--th iteration of the parameters $ \theta_i^{(h+1)} $ using the KKT conditions of problem \eqref{eq_estimation}, presented in Lemma~\ref{lemma_KKT} of Appendix~A. More specifically, 
\begin{equation}\label{eq_bcd}
{\theta_i^{\prime}}^{(h+1)} = { \left( \sum_{l=i}^{n} Y_{l-1} Y_{l-1}^\prime \right) }^{-1} S \left( \sum_{l=i}^{n} Y_{l-1} y_{l}^{\prime} - \sum_{j \ne i} \left( \sum_{l = \max(i,j)}^{n} Y_{l-1} Y_{l-1}^\prime \right) {\theta_j^\prime}^{(h)} ; \lambda_{1,n} \right),
\end{equation}
where $ Y_l^\prime = \left( y_l^\prime \ldots y_{l-q+1}^\prime  \right)_{1 \times pq} $.
\item[(iii)]
\begin{itemize}
\item[(a)] If $ \max_{1 \leq i \leq n} \| {\theta_i}^{(h+1)} - {\theta_i}^{(h)} \|_\infty < \delta$, where $\delta$ is the tolerance set to $ 10^{-3} $ in our implementation, stop the iteration and denote the final estimate by $ \Theta^{(intermediate)} $.
\item[(b)] If $ \max_{1 \leq i \leq n} \| {\theta_i}^{(h+1)} - {\theta_i}^{(h)} \|_\infty \geq \delta$, set $ h = h + 1 $. Go to step (ii).
\end{itemize}
\item[(iv)] Apply soft--thresholding to $ \Theta^{(intermediate)} $ to find the optimizer in equation \eqref{eq_estimation}. In other words, $ \widehat{\Theta} = S (\Theta^{(intermediate)}; \lambda_{2,n}) $.
\end{itemize} 

Note that in this algorithm, the whole block of $ \theta_i$ with $ p^2 q $ elements is updated at once, which reduces the computation time dramatically.

Our backward elimination algorithm (BEA) for the second-stage screening is as follows:

\begin{itemize}
\item[(i)] Set $ m = | \widehat{\mathcal{A}}_n | $. Let $ \textbf{s} = \lbrace  s_1, 
\ldots, s_m \rbrace $ be the selected points and define $ W_m^\star = \mathrm{IC}(s_1, \ldots, s_m; \eta_n) $.
\item[(ii)] For each $ i = 1, \ldots, m $, calculate $ W_{m,i} = \mathrm{IC}( \textbf{s} \backslash \lbrace s_i \rbrace; \eta_n )  $. Define $ W_{m-1}^\star = \min_{i}  W_{m,i} $. 
\item[(iii)]
\begin{itemize}
\item[(a)] If $ W_{m-1}^\star >  W_m^\star $, then no further reduction is needed. Return $\widehat{\mathcal{A}}_n$ as the estimated change points.
\item[(b)] If $ W_{m-1}^\star \leq  W_m^\star $, and $ m > 1 $, set $ j = \mbox{argmin}_{i} W_{m,i} $, set $ \textbf{s} = \textbf{s} \backslash \lbrace s_j \rbrace $ and $ m = m - 1 $. Go to step (ii).
\item[(c)] If $ W_{m-1}^\star \leq  W_m^\star $ and $ m = 1 $, all selected points are removed. Return the empty set. 
\end{itemize}
\end{itemize}

\subsection*{Appendix~D: Additional Simulation Results}

In this section, two additional simulation scenarios are described and the empirical results are reported.

\emph{Simulation Scenario 4 (Randomly structured $ \Phi $ and break points close to the center)}. As in Scenario 1, in this case we set $ t_1 = 100 $ and $ t_2 = 200 $. However, the coefficients matrices are chosen to be randomly structured.  {
The autoregressive coefficients for simulation scenarios 1 and 2 are displayed in Figure~\ref{fig_phi}. The 1-off diagonal values for the three segments are -0.6, 0.75, and -0.8, respectively.}
However, the autoregressive coefficients for this scenario are chosen to be randomly structured as displayed in Figure~\ref{fig_phi_2}.

The selected break points in this scenario are shown in the middle part of Figure~\ref{fig_sim_3}. The mean and standard deviation of locations of the selected break points, relative to the sample size $ T $, as well as the percentage of simulation runs where break points are correctly identified are shown in Table~\ref{table_sim_2}. The results suggest that, among all simulation scenarios, this setting, with randomly structured $ \Phi$'s, is the most challenging for our method in terms of detecting the number of break points. In this setting, the detection rate drops to $ 99\% $ compared to $100\%$ in the previous scenarios, and the standard deviation of the selected break point locations are higher than the first scenario. The percentage of runs where true break points are within $ R_n$-radius of the estimated points also drops to $ 96\% $ compared to $100\%$ in Scenarios 1, 2 and 3.

\begin{figure}[h]
\begin{center}
\includegraphics[width=0.5\textwidth, clip=TRUE, trim=0mm 25mm 0mm 30mm]{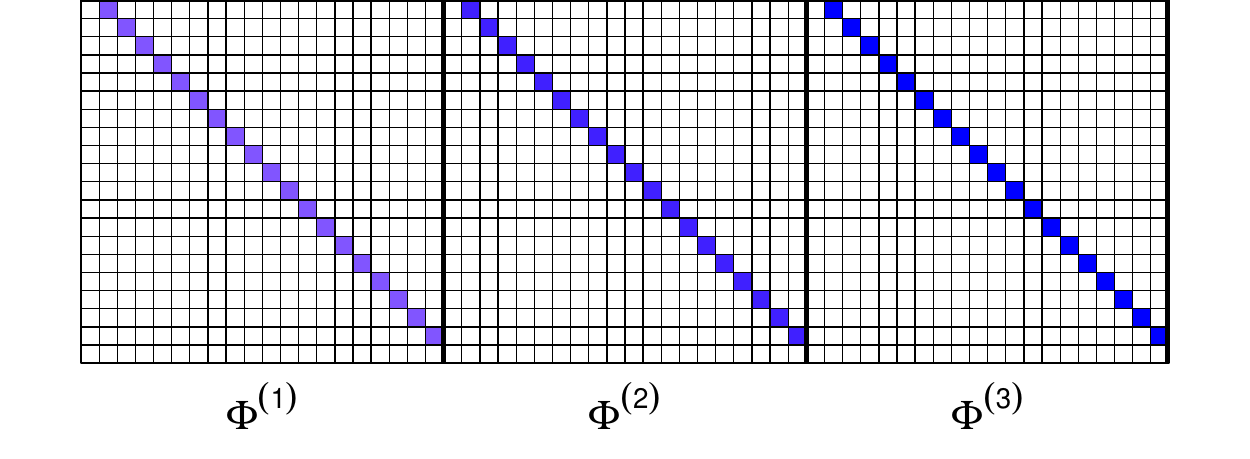}
\caption{True autoregressive coefficients for the three segments in Simulation Scenarios 1 and 2.}
\label{fig_phi}
\end{center}
\end{figure}

\begin{figure}[t!]
\begin{center}
\includegraphics[width=0.5\textwidth, clip=TRUE, trim=0mm 25mm 0mm 30mm]{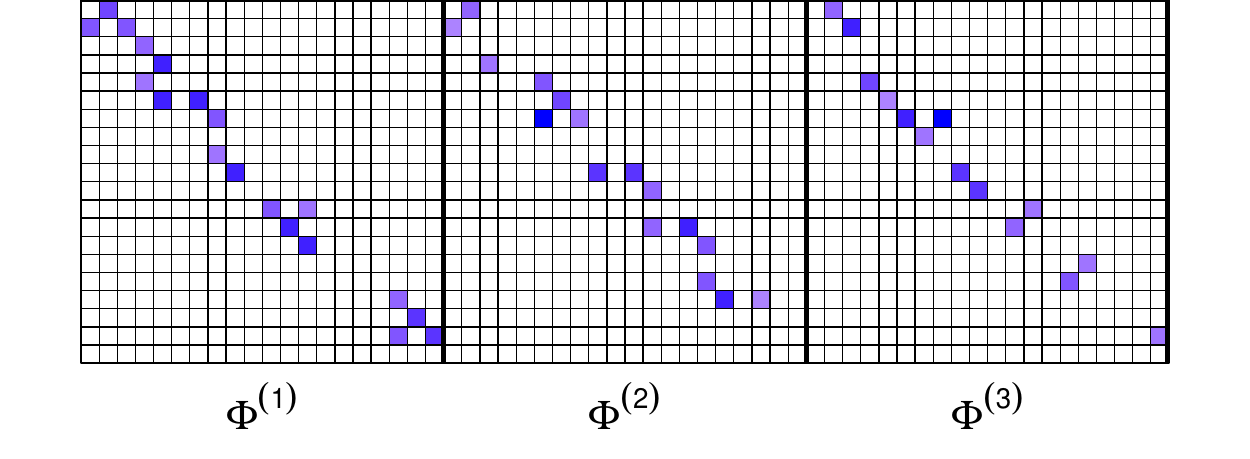}
\caption{True autoregressive coefficients for the three segments in Simulation Scenario 4.}\label{fig_phi_2}
\end{center}
\end{figure}

The inferior performance of the proposed method in the fourth simulation scenario could be due to the fact that the $ L_2 $-distance between the consecutive autoregressive coefficients are less than the previous two cases. {
The $ L_2 $ norm of the consecutive differences of the VAR parameters in simulation 1 and 2 are 5.88 and 6.76 whereas in simulation 4, they are 4.44 and 4.49. This $ 35\% $ reduction in the $ L_2 $-distance between the consecutive autoregressive coefficients }would make it harder to identify the exact location of the break points. In contrast, the sparse changes in coefficient matrices makes this setting more favorable for SBS-MVTS. Nonetheless, estimates from our method are as good or better than those from SBS-MVTS.

Table~\ref{table:sim4_AR} summarizes the results for autoregressive parameter estimation in this simulation scenarios. The table shows mean and standard deviation of relative estimation error (REE), as well as true positive (TPR) and false positive rates (FPR) of the estimates. The results suggest that the proposed method performs well in terms of parameter estimation. {
In simulation scenario 4, the performance of SBS--MVTS is better in estimation and in true positive rate. One reason for this good performance is that in this scenario, the selected break points of SBS--MVTS method are close enough to the true break points which makes it unnecessary to remove the $ R_n$--radius of them in order to ensure stationarity. However, in real data applications, since the ground truth is unknown, this removal becomes necessary.  }

\begin{table}
\hspace{-0.5cm}
\caption{\label{table_sim_2} Results for Simulation Scenario 4. The table shows mean and standard deviation of estimated break point locations, the percentage of simulation runs where break points are correctly detected (selection rate), and the percentage of simulation runs where true break points are within the $ R_n$-radius of the estimated break points ($ R_n$-selection rate).}
\centering
\begin{tabular}{lcccccc}
  \hline
method & break points & truth & mean & std & selection rate & $R_n$-selection rate \\ 
  \hline
    \hline
  SBS-MVTS & 1 & 0.3333 & 0.3238 & 0.0206 & 0.98 & -- \\ 
  \, & 2 & 0.6667 & 0.6569 & 0.0324 & 0.92 & -- \\
  Our method  & 1 & 0.3333 & 0.3323 & 0.0124 & 0.99 & 0.98 \\ 
  \, & 2 & 0.6667 & 0.6620 & 0.0200 & 0.99 & 0.96 \\  
   \hline
\end{tabular}
\end{table}

\begin{table}
\hspace{-1cm}
\caption{\label{table:sim4_AR} Results of parameter estimation for simulation scenario 4. The table shows mean and standard deviation of relative estimation error (REE), true positive rate (TPR), and false positive rate (FPR) for estimated coefficients.}
\centering
\begin{tabular}{lccccc} 
  \hline
 & Method & REE & SD(REE) & TPR & FPR \\
  \hline
    \hline 
  & Our Method   & 0.5263 & 0.0558 & 0.94 & 0.03 \\ 
  & SBS-MVTS     & 0.2757 & 0.1099 & 1.00 & 0.04  \\ 
   \hline
\end{tabular}

\end{table}

\begin{figure}[t!]
\begin{center}
\includegraphics[width=0.32\textwidth, height=0.15\textheight, clip=TRUE, trim=0cm 0cm 0cm 2cm]{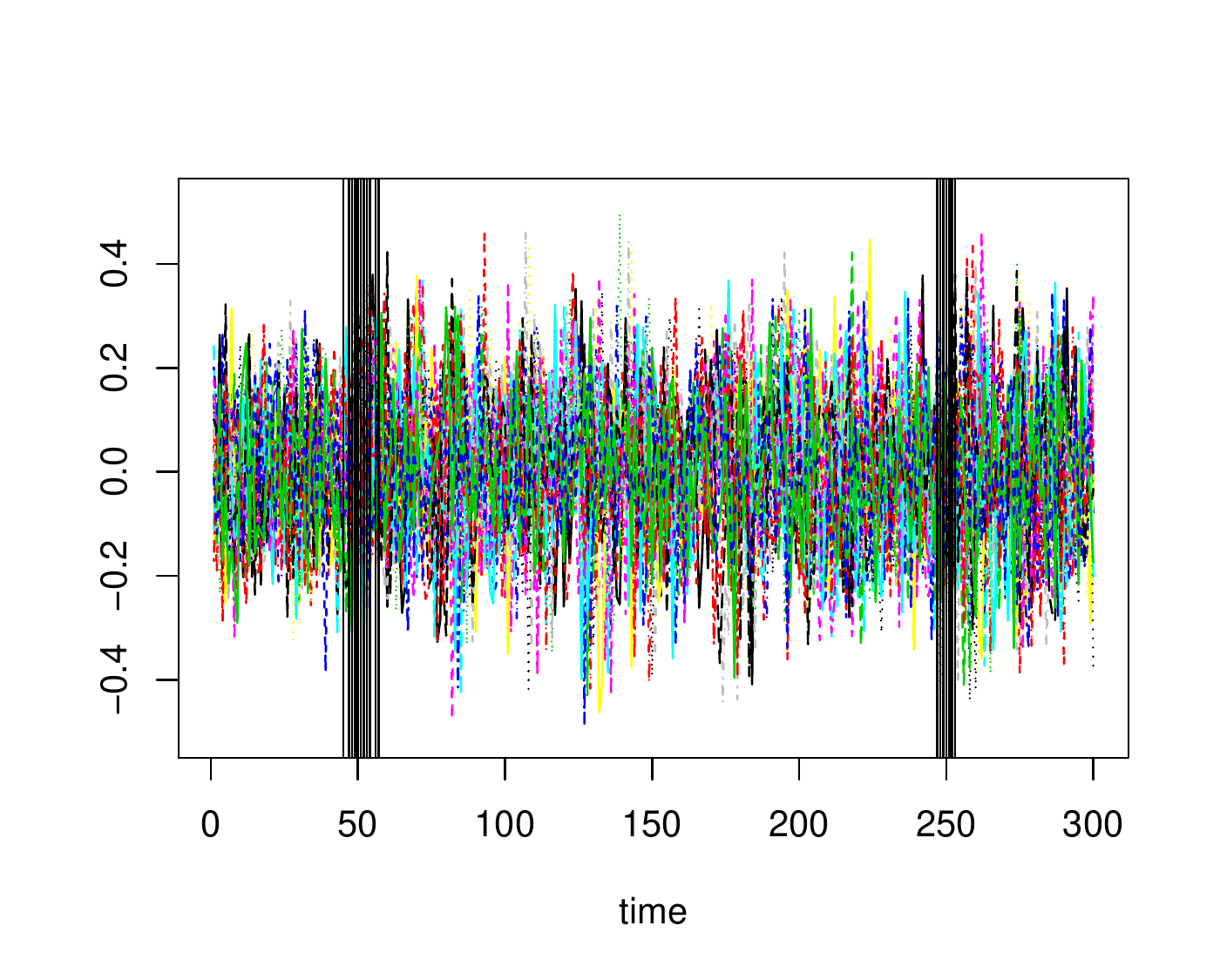}
\includegraphics[width=0.32\textwidth, height=0.15\textheight, clip=TRUE, trim=0cm 0cm 0cm 2cm]{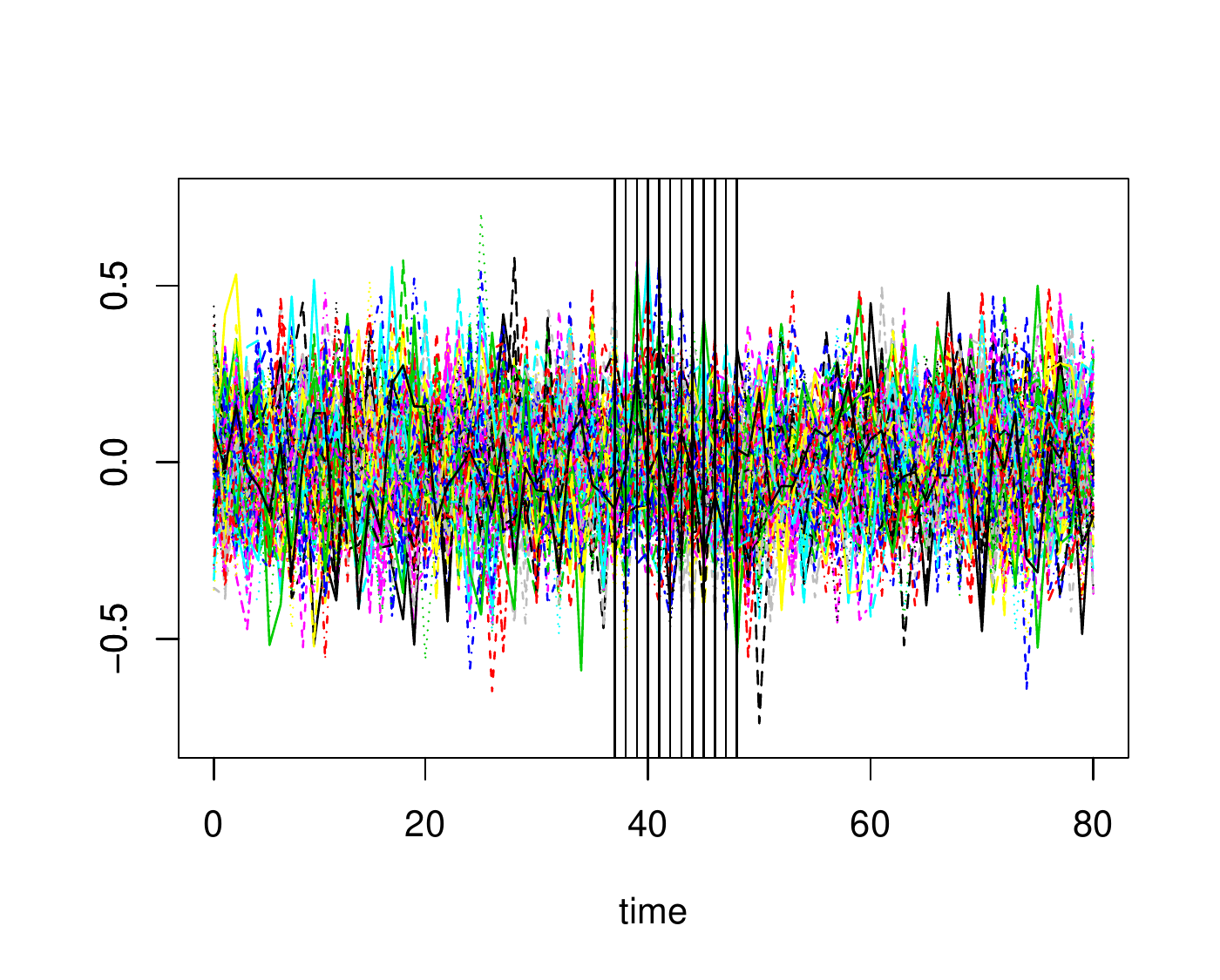}
\includegraphics[width=0.32\textwidth, height=0.15\textheight, clip=TRUE, trim=0cm 0cm 0cm 2cm]{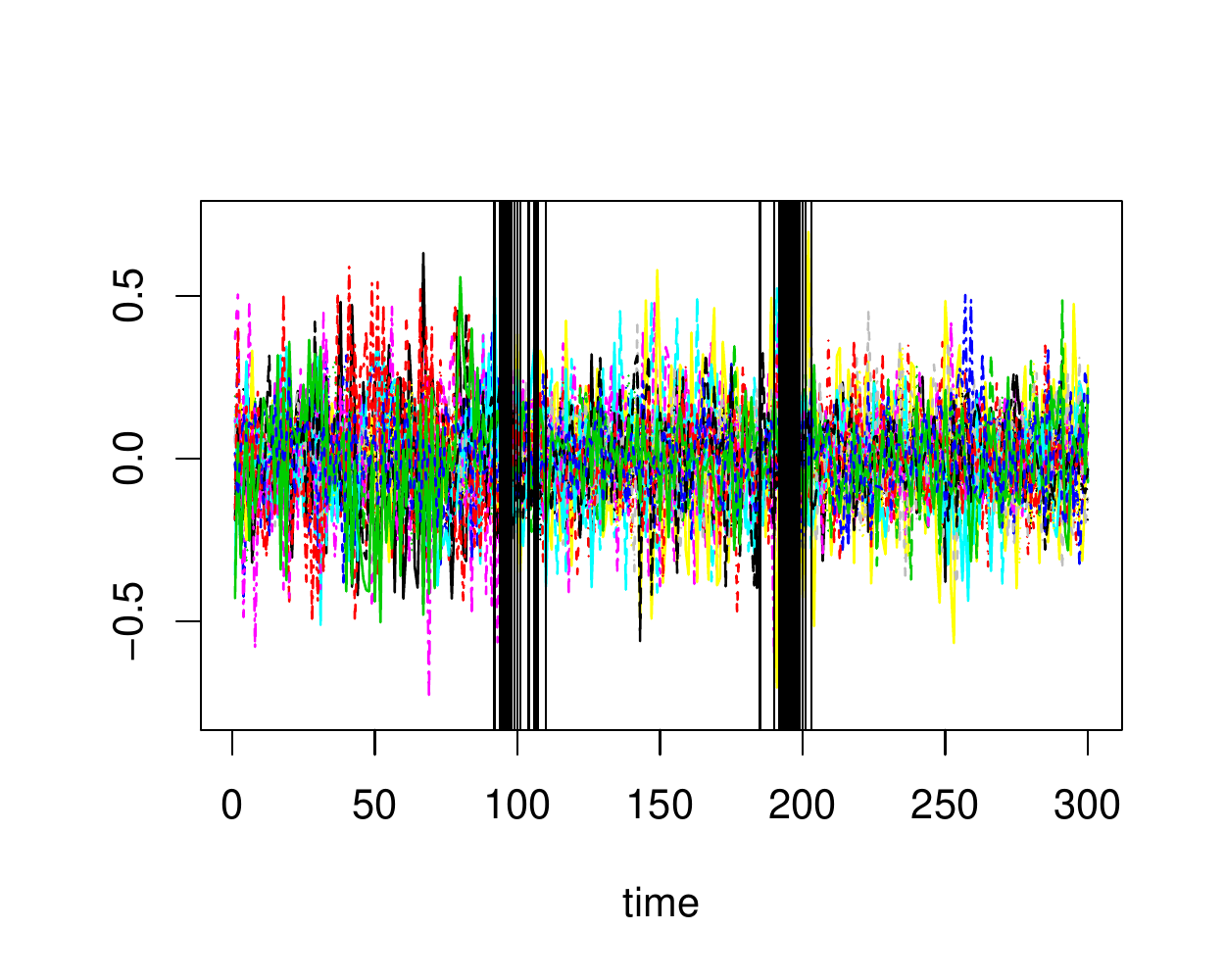}
\vspace{-0.5cm}
\caption{Estimated break points using our method for all 100 runs from Simulation Scenarios 2 (left), 3 (middle) and 4 (right).}\label{fig_sim_3}
\end{center}
\end{figure}

\emph{Simulation Scenario 5 (Simple $ \Phi $ and correlated error term)}. As in Scenario 1, in this case we set $ t_1 = 100 $ and $ t_2 = 200 $ with $ T = 300 $. The coefficients matrices are chosen to be the same as in simulation scenario~1 as displayed in Figure~\ref{fig_phi}. However, the covariance matrix of the error terms is dense. More specifically, $ \Sigma_{\varepsilon} = 0.01 \, (( \sigma_{ij} ))_{T \times T} $ with $ \sigma_{ij}  = {0.5}^{|i-j|} $. The reason to add this simulation scenario is to see the effect of additional correlation structure of the noise term on the performance of our method both in terms of detection and estimation. 

\begin{table}
\hspace{-0.5cm}
\caption{\label{table_sim_5} Results for Simulation Scenario 5. The table shows mean and standard deviation of estimated break point locations, the percentage of simulation runs where break points are correctly detected (selection rate), and the percentage of simulation runs where true break points are within the $ R_n$-radius of the estimated break points ($ R_n$-selection rate).}
\centering
\begin{tabular}{lcccccc}
  \hline
method & break points & truth & mean & std & selection rate & $R_n$-selection rate \\ 
  \hline
    \hline
  SBS-MVTS & 1 & 0.3333 & 0.3688 & 0.0413 & 0.68 & -- \\ 
              \, & 2 & 0.6667 & 0.6119 & 0.0945 & 0.80 & -- \\
  Our method  & 1 & 0.3333 & 0.3251& 0.0139 & 1 & 1\\ 
                  \, & 2 & 0.6667 & 0.6507 & 0.0213 & 1 & 1\\  
   \hline
\end{tabular}
\end{table}

Table~\ref{table_sim_5} reports the performance of our method and the SBS-MVTS developed in \citep{cho2015multiple} in the simulation scenario~5 in terms of detection of break points. More specifically, the mean and standard deviation of locations of the selected break points, relative to the sample size $ T $, as well as the percentage of simulation runs where break points are correctly identified are shown in Table~\ref{table_sim_5}. As seen from this table, our method performs very well in this scenario which confirms the applicability of our method in the case of correlated error terms.  

\begin{table}
\hspace{-1cm}
\caption{\label{table:sim5_AR} Results of parameter estimation for simulation scenario 5. The table shows mean and standard deviation of relative estimation error (REE), true positive rate (TPR), and false positive rate (FPR) for estimated coefficients.}
\centering
\begin{tabular}{lccccc} 
  \hline
 & Method & REE & SD(REE) & TPR & FPR \\
  \hline
    \hline 
  & Our Method   & 0.6012 & 0.0699 & 0.93 & 0.04 \\ 
  & SBS-MVTS     & 0.8005 & 0.1693 & 0.70 & 0.01  \\ 
   \hline
\end{tabular}

\end{table}

Table~\ref{table:sim5_AR} summarizes the results for autoregressive parameter estimation in the simulation scenario~5. The table shows mean and standard deviation of relative estimation error (REE), as well as true positive (TPR) and false positive rates (FPR) of the estimates. The results suggest that the proposed method performs well in terms of parameter estimation and is superior to the naive approach using the detected points of SBS-MVTS method and applying regularization method for parameter estimation. 

%

\end{document}